\documentclass{article}
\usepackage{graphicx}
\usepackage{amsmath}
\usepackage{amssymb}
\usepackage{amsthm}
\newtheorem{definition}{Definition}
\newtheorem{theorem}{Theorem}
\usepackage{appendix}
\usepackage{booktabs}
\usepackage{float}
\usepackage{array}
\usepackage{algorithm,algpseudocode}
\usepackage{subcaption}
\usepackage{url}

\title{Resilient Liquid Democracy:\\ Mitigating Voting Power Imbalances via Secure Delegation Networks}
\author{Zhuolun Li$^{1}$, Evangelos Pournaras$^{1,2}$ \\
$^{1}$School of Computer Science, University of Leeds, UK \\
$^{2}$School of Energy Systems, LUT University, \\ 53850 Lappeenranta, Finland}

\begin{document}

\maketitle

\begingroup
\renewcommand\thefootnote{}
\footnotetext{This work has been submitted to the IEEE for possible publication. Copyright may be transferred without notice, after which this version may no longer be accessible.}
\addtocounter{footnote}{-1}
\endgroup

\begin{abstract}
Liquid democracy lets voters either vote directly or delegate their voting power to another trusted participant. It is often deployed in political parties, municipal participation platforms, and the governance of decentralized autonomous organizations. Existing approaches make delegations publicly visible as they form, which results in herding driven by popularity, makes coercion verifiable, and leaves the election fragile when popular delegates abstain. We propose a novel liquid democracy mechanism that removes these vulnerabilities, while keeping the tally fully auditable. Delegation choices are sealed with decentralized timed-release encryption while they are formed, and revealed in full for tallying. Each voter determines ranked backup delegates together with a personal fallback ballot so that delegation failures do not result in abstaining. Experiments grounded in real-world voting data show that the mechanism keeps voting power dispersed, whereas transparent formation concentrates it into few hands. Strikingly, vote loss under targeted delegate failures falls from $26\%$ to $3\%$. 
\end{abstract}

\section{Introduction}
\label{sec:intro}

In real-world collective decision making, participants often differ in how much time, information, and confidence they have to make an informed decision. Some participants have substantial domain knowledge, some have relevant local or professional experience, and others simply trust certain people more than they trust themselves on a given decision. 

Liquid democracy is often proposed as a mechanism for such settings. It allows a voter to transfer their voting power to another participant. This flexibility can lower the burden of participation when a voter is busy, uncertain, or insufficiently informed. It can also improve decision quality when voters choose delegates whom they regard as more reliable, more informed, or more aligned with their judgment.

Despite this potential, existing liquid democracy systems are poorly suited for secure and reliable use. Most implementations operate under a transparent delegation regime, in which delegation choices are publicly observable as they are formed. This harms the voting outcome in three ways. First, visible delegations encourage herding. Early or well-known delegates attract support because they already appear popular. Randomized experiments show that visible popularity signals causally attract further endorsements and that social signals shift real political participation~\cite{muchnik2013social, bond2012sixtymillion}. Voting power then concentrates in a few delegates. In a liquid democracy platform LiquidFeedback, received delegations became heavily concentrated in a few super voters~\cite{kling2015voting}, and in decentralized autonomous organizations (DAOs) a small subset of delegates controls the vast majority of voting power~\cite{fritsch2024analyzing, schmid2024blockchain, weidener2025delegated}. Second, transparent delegation makes influence tradable. Vote buyers or coercers can confirm how voters delegated and condition payment or retaliation on it. Third, when a heavily-backed delegate abstains or becomes inactive, every vote routed to them is lost at once, and the outcome shifts toward whoever remains (Figure~\ref{fig:scenario}a). 

Existing delegated voting systems show what happens when a few identifiable actors hold decisive power. Members of the 2016 US Electoral College received organized pressure campaigns and death threats intended to change or block their votes, and two resigned instead of casting them~\cite{pbs2016electors}. In the 2020 US election, senior figures of the losing campaign personally telephoned the two pivotal members of a county canvassing board, who afterwards attempted to rescind their certification votes~\cite{fox2detroit2023canvassers}. The top-ranked witnesses of the Steem blockchain were collectively removed in a hostile takeover~\cite{dale2020steem}. Delegated voting power on Arbitrum is openly purchasable per proposal through lobbying markets~\cite{matos2025lobbyfi, austgen2023dao}. Liquid democracy platforms share the precondition for such interference, since the delegates they rely on most are publicly visible and hold concentrated power.

We address these problems by designing a liquid democracy mechanism for secure trusted delegation, based on two ideas. First, we introduce sealed delegation using timed-release cryptography to prevent participants and external observers from learning who delegated to whom when decisions are being formed. This limits herding and removes the evidentiary basis that bribery and coercion strategies rely on during the delegation formation phase. Second, we extend standard liquid democracy with ranked multi-delegation and fallback ballots. Each voter may specify an ordered list of preferred delegates, together with a personal backup ballot that is used if all delegations fail, so that votes are preserved even when delegates abstain or delegation cycles occur (Figure~\ref{fig:scenario}b).

\begin{figure}[!t]
\centering
\includegraphics[width=\textwidth]{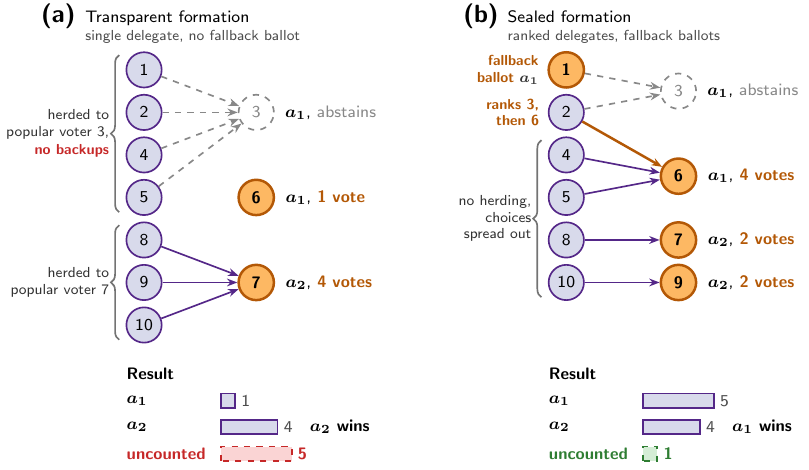}
\caption{Ten voters choose between $a_1$ and $a_2$. Six prefer $a_1$, and voter 3 abstains. Arrows are delegations, dashed arrows carry no votes, and orange voters cast the counted ballots. (a) Delegations concentrate on voter 3 without backups, so five ballots are lost and $a_2$ wins. (b) Delegations are spread, voter 2 moves to their second-ranked delegate, and voter 1 uses their fallback ballot, $a_1$ wins.}
\label{fig:scenario}
\end{figure}

This paper makes the following contributions:
\begin{itemize}
    \item We formalize liquid democracy as a mechanism for trusted delegation and introduce the notion of delegation information regimes. An information regime captures what voters can observe about the delegation choices of others during decision formation. We study experimentally how this visibility shapes delegation behavior and network structure.
    \item We design a liquid democracy protocol that combines timed-release cryptography with ranked multi-delegation and fallback ballots. It resists herding, bribery, and coercion during decision formation and minimizes vote loss under delegate failure.
    \item We give a formal security analysis of the mechanism. We prove that no efficient attacker learns the delegation of a voter before the reveal time, and that a coerced voter can comply with a demand yet still cast their true preference undetectably before the reveal time.
    \item We show that sealed formation produces a less concentrated and more robust delegation network. Ranked multi-delegation with fallback ballots holds vote loss at a floor of about $3\%$ under targeted delegate failures, while transparent formation with a single delegate loses up to $49\%$ of the votes.
    \item We also show how the representation scheme of the delegate determines whether delegation makes voting outcomes more accurate. A delegate casting the mean ballot of their constituency recovers the most accuracy on point ballots.
\end{itemize}

In the remaining paper, Section~\ref{sec:related} reviews work on delegation as a way to aggregate expertise, the risks observed in deployed delegation platforms, and cryptographic voting protocols that protect ballot privacy. Section~\ref{sec:formalization} formalizes liquid democracy, defines how delegations are resolved into counted ballots, and introduces the transparent and sealed information regimes. Section~\ref{sec:design} presents the proposed mechanism. It seals delegation choices with timed-release encryption during the voting phase, resolves them in the reveal phase, and states its security guarantees. Section~\ref{sec:experiments} describes the datasets, the calibration of voter expertise and delegate selection, the accuracy metric, and the model of delegate failure. Section~\ref{sec:evaluation} reports vote loss under delegate failure, the concentration of voting power under each information regime, and the effect of the representation scheme on outcome accuracy. Section~\ref{sec:conclusion} concludes the paper and outlines future work. 

\section{Related Work}
\label{sec:related}

Liquid democracy builds on a broader literature studying how individuals transfer decision power to others. Theoretical work on delegation, the record of deployed delegation platforms, and cryptographic voting each inform a different part of our design.

\subsection{Delegation as a Mechanism for Expertise Aggregation}

The paradigm originates in proposals for delegative or proxy democracy that combine direct and representative voting through freely chosen, transitive delegation~\cite{ford2002delegative, green2015direct, blum2016liquid}. In epistemic models of collective decision making, delegation is motivated by differences in competence or information across participants, so that decisions can be routed toward people with higher expected decision quality~\cite{green2015direct, blum2016liquid}. A substantial theoretical literature asks whether delegation actually improves decision accuracy and finds that it helps only conditionally. Delegation mechanisms are not guaranteed to outperform direct voting~\cite{kahng2021liquid, caragiannis2019contribution}. Accuracy can fall when a few delegates collect most of the delegated votes~\cite{golz2018fluid}. When voters choose delegates independently, with only a mild preference for more competent ones, this concentration stays limited and delegation is at least as accurate as direct voting~\cite{halpern2023defense}. Formal analyses of voting power quantify how delegation redistributes influence~\cite{zhang2021power, colley2023measuring}, and a survey organizes this literature through a taxonomy of modeling choices~\cite{grossi2026modeling}. Our empirical result across datasets (Section~\ref{sec: delegation_model}) is a data-driven counterpart to this debate, identifying when delegation helps and when it actively harms. Algorithmic and game-theoretic work also studies how rational agents choose between direct voting and delegation, and how the resulting delegation graph affects the final outcome~\cite{bloembergen2019rational, escoffier2020iterative, colley2021multi}.

Empirical and experimental studies show that voters do not always delegate in an idealized way~\cite{campbell2022liquid, gersbach2021vote}. Participants may rely on limited information and misjudge who is an appropriate delegate. These findings suggest that delegation mechanisms must be robust not only to strategic manipulation, but also to ordinary social influence and imperfect judgment. Our delegation model adopts this lesson, selecting delegates by demographic and ideological similarity and modeling the social influence channel explicitly, as illustrated in Section~\ref{sec:experiments}. 

\subsection{Risks in Delegation Networks and Deployed Systems}

A key challenge in liquid democracy is that the delegation network itself can become a target of manipulation. Alouf-Heffetz et al.~\cite{alouf2024controlling} analyze how changes to delegation edges can alter outcomes. Votes can also be lost without any attacker. A delegation cycle arises when a chain of delegations loops back on itself, so that no voter in the chain reaches a direct ballot. Such a cycle can void every vote in the chain, and so can a delegate who simply abstains. Iterative processes, in which voters keep reassigning their delegations in response to each other, may also never settle~\cite{dey2021parameterized, kotsialou2018incentivising}. Several proposals strengthen the mechanism against such failures. Ranked delegation allows a voter to list backup delegates~\cite{brill2022liquid, utke2023anonymous}, and generalized models allow splitting influence across several delegates, delegating to a group, or discounting long delegation chains~\cite{bersetche2022generalizing, colley2021multi}. These approaches improve reliability after delegations are formed, but they assume that delegation choices are observable during formation.

Deployed systems show why this observability is harmful. Liquid democracy has been used in civic platforms and party decision making. Examples include LiquidFeedback and Liquid Friesland~\cite{liquidfeedback, liquid-friesland}, the Swedish Direktdemokraterna~\cite{direktdemokraterna}, and a corporate deployment inside Google that displayed voting and delegation activity to participants~\cite{hardt2015googlevotes}. Governance by delegation is now common in blockchain ecosystems, where DAOs let token holders delegate to representatives through tools such as Tally~\cite{tally}. Across these systems, delegated power concentrates. In LiquidFeedback, delegations are visible to members. The distribution of received delegations is heavy-tailed, and concentration rises over time as super voters emerge~\cite{kling2015voting}. In DAOs, a small subset of delegates controls the vast majority of voting power~\cite{fritsch2024analyzing, schmid2024blockchain, weidener2025delegated}.

Visible concentration reinforces itself, since randomized experiments show that visible popularity signals causally attract further endorsements~\cite{muchnik2013social, vanderijt2014field, glenski2017rating}, and it exposes the targeted interference to the popular delegates documented in Section~\ref{sec:intro}. The failure model of Section~\ref{sec: failure_model} treats targeted delegate failure as an anticipated threat.

\subsection{Privacy-Preserving Voting and Delegation}

A separate body of work studies how to protect voter privacy and ensure secure vote execution. The secure e-voting literature targets receipt-freeness and coercion-resistance. A receipt is evidence that a voter can hand to a coercer to prove how they voted, and receipt-freeness means that a voter cannot produce such evidence even if they wish to~\cite{benaloh1994receipt, juels2005coercion}. The literature also targets end-to-end verifiability, so that anyone can audit the tally~\cite{adida2008helios, clarkson2008civitas}. For delegation, Link et al. give each voter an anonymous voting credential, which the voter passes privately to their delegate~\cite{link2014opportunistic}. Neither the system nor the public ever learns who delegated to whom, so the delegations cannot be audited after the vote. Ballot secrecy has also been studied for liquid democracy directly~\cite{nejadgholi2021short}, and anti-collusion infrastructure for on-chain voting pursues receipt-freeness in the setting we study~\cite{maci}. A common practical defense against coercion is to permit re-voting and count only the last ballot of a voter, as deployed in Estonian Internet voting. The security analysis of that system, however, also documents the limits of re-voting as a defense, including the fact that re-voting itself is observable~\cite{springall2014estonian}. Our receipt-freeness via resubmission (Section~\ref{sec: receiptfree}) builds on this re-voting idea but couples it with sealed formation, so that a coercer cannot even tell whether a resubmission changed the vote.

Most privacy-preserving systems aim to hide ballots or identities, whereas our concern is the timing of delegation visibility. In particular, we ask whether delegation choices should be visible while voters are still making them. That question is central to herding, coercion, and popularity effects. It is not the same question as standard ballot secrecy, and our design in fact trades away post-reveal ballot secrecy to obtain a fully auditable tally, a trade-off we return to in Section~\ref{sec:security}.

Table~\ref{tab:comparison} compares prior systems on (i) whether delegation choices are visible during formation, (ii) whether ranked backup delegates and a fallback ballot are resilient against delegate failure, and (iii) whether the final tally is publicly auditable. No prior line of work combines all three. Transparent systems provide auditability but expose delegations during decision formation. Privacy-preserving systems hide delegations permanently, so nobody can audit how they were resolved, and they offer no fallback ballot.

\begin{table}[t]
\centering
\small
\caption{Positioning of this work relative to prior liquid democracy systems and mechanisms.}
\label{tab:comparison}
\vspace{5pt}
\resizebox{\textwidth}{!}{
\begin{tabular}{p{3.1cm}p{2.8cm}p{1.7cm}p{1.6cm}p{1.9cm}}
\toprule
System / line of work & Delegation hidden when casting ballots & Ranked delegates & Fallback ballot & Publicly auditable tally \\
\midrule
LiquidFeedback / Liquid Friesland / Tally~\cite{liquidfeedback, liquid-friesland, tally} & No & Limited & No & Yes \\ \hline
Ranked delegation models~\cite{brill2022liquid, utke2023anonymous} & No & Yes & No & Yes \\ \hline
Anonymous credential delegation~\cite{link2014opportunistic} & Yes & Limited & No & Limited \\ \hline 
This work & Yes & Yes & Yes & Yes \\
\bottomrule
\end{tabular}
}
\end{table}

The question of when a decision is better handed to someone else also arises in machine learning and artificial intelligence within the paradigm of learning to defer. A model learns whether to predict or to defer to a human expert for safety in high-risk decision making~\cite{madras2018predict, mozannar2020consistent}. Extensions of the paradigm choose among several experts~\cite{verma2023learning} and cap the load on each~\cite{alves2024cost}. Deferral helps only when the expert is better on the deferred cases, and no single expert should undertake every decision. Our results on delegation accuracy and power concentration mirror both conditions. The delegate may also be a machine. Majumdar et al.\ let large language model personas represent abstaining voters~\cite{majumdar2024generative}. Such delegation sharpens the risks studied here, since several voters may delegate to the same model and one failure or bias then reaches every vote routed through it. A mechanism that keeps delegated power dispersed and reroutes around failed delegates is therefore a precondition for trustworthy delegation to AI agents, and ours applies to human and machine delegates alike.

Prior work and deployed systems thus reveal two vulnerabilities~\cite{kahng2021liquid, fritsch2024analyzing}. Transparent delegation exposes voters to popularity signals that trigger herding, and delegation concentrates voting power in a few delegates whose failure can disenfranchise large portions of the electorate. Our design addresses both.

\section{Formalization of Liquid Democracy}
\label{sec:formalization}

We formalize liquid democracy as a mechanism for trusted delegation in collective decision making. Participants may either vote directly or route their voting power toward another participant whom they regard as more reliable, more informed, or better placed to judge for the decision.
The notation used in the formalization, the system design, the security analysis, and the experiments is collected in Table~\ref{tab:notation} in Appendix~\ref{app:notation}.

\subsection{Participants, Identity, and Expert Delegation Model}

Let $N = \{1,2,\dots,n\}$ denote the finite set of voters, and let $A = \{a_1, a_2, \dots, a_m\}$ denote the set of policy alternatives or candidates. Participants interact with the system through pseudonymous identities such as blockchain addresses or public keys. Ballots and delegation relationships are associated only with these pseudonymous identifiers.

Each voter $i \in N$ submits a ballot that can contain two elements, a direct ballot $b_i$ over the alternatives in $A$ and a ranked list of chosen delegates. Depending on what the voter provides, this allows for:

\begin{enumerate}
    \item \textbf{Direct vote:} The voter provides only a direct ballot $b_i$, which is counted directly.
    \item \textbf{Delegation:} The voter leaves the direct vote blank and provides only a list of delegates. Their voting power is transferred.
    \item \textbf{Delegation with fallback:} The voter provides both a list of delegates and a direct ballot $b_i$. The direct vote acts as a fallback, which is only counted if the chosen delegates abstain, fail to vote, or form a delegation cycle that cannot be routed to a direct voter.
\end{enumerate}

\subsection{Delegation Graph and Vote Resolution}
\label{sec: resolution_def}

After ballots are submitted, the system resolves delegations. Resolution is the process that follows the chain of delegations for each voter until it reaches a participant who casts a direct ballot, so that the power of every voter ends up assigned to exactly one counted ballot. The delegation relationships form a directed graph, called the delegation graph, with an edge from voter $i$ to voter $j$ if $i$ delegates to $j$.

Delegation satisfies the following properties:

\begin{itemize}
    \item \textbf{Transitivity.} If $i$ delegates to $j$ and $j$ delegates to $k$, then the vote of $i$ follows the delegation path until it reaches a voter who casts a direct ballot.
    \item \textbf{Cycles.} Delegation cycles may occur and must be resolved by system rules such as abstention, cycle breaking, or fallback voting.
\end{itemize}

A voter whose effective choice after resolution is a direct ballot is a terminal voter, that is, a delegate or direct voter at whom one or more delegation paths terminate. For each terminal voter $g$, define their voting weight as
\[
W(g) = \big|\{\, i \in N : \text{the delegation path of $i$ terminates at $g$}\,\}\big|.
\]
Thus $W(g)$ is the number of voters whose voting power is ultimately routed to $g$; equivalently, the ballot of $g$ is counted $W(g)$ times in the final tally.
In the example of Figure~\ref{fig:resolution}, the delegation paths of Voters 1, 2, 4, and 5 end at Voter 3, so $W(3)=5$. The cycle between Voters 4 and 5 is broken by advancing Voter 4 to their second-ranked delegate, following the rule of Section~\ref{sec:reveal}. Voter 6 falls back to their own ballot because Voter 7 is inactive, and Voter 8 loses their vote because they rank only Voter 7 and supplies no fallback ballot. Vote loss, the fraction of the electorate whose ballot is not counted, is two of eight in this example.

\begin{figure}[h!]
\centering
\includegraphics[width=0.8\textwidth]{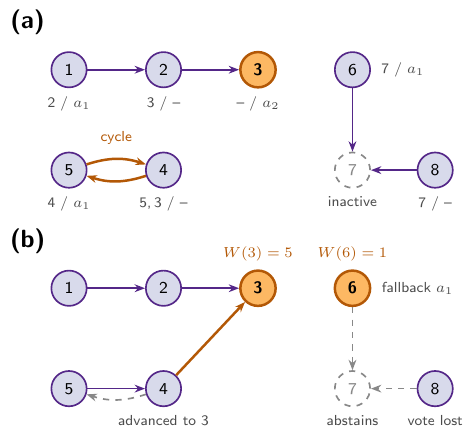}
\caption{Delegation resolution on eight voters. (a) The delegation graph from the first-ranked choices. The label under each voter lists their ranked delegates and, after the slash, their fallback ballot, with a dash for none. (b) The graph after resolution. Orange nodes are terminal voters with their weights $W(g)$, dashed grey edges carry no voting power, and voter 7 is inactive. Six of the eight ballots are counted, and the two uncounted ones belong to the inactive voter and to the delegator without a fallback.}
\label{fig:resolution}
\end{figure}

\paragraph{Representation schemes.} \label{par: rep_del_mod} The constituency of a terminal voter $g$ consists of $g$ and every voter whose delegation path ends at $g$. The representation scheme of $g$ decides which ballot is counted $W(g)$ times for this constituency. Under a personal representation scheme, $g$ votes on their own preferences. Under a representative scheme, $g$ casts a ballot that aggregates the preferences of the whole constituency. When delegations are sealed, as in the mechanism of Section~\ref{sec:design}, a delegate does not learn who delegated to them before voting closes. The mechanism therefore leaves open how a representative delegate would learn the preferences of their constituency, and we treat the representative scheme as a modeling assumption. Section~\ref{sec: rep_style} uses it to measure how much accuracy a delegate could recover if they knew these preferences.

After resolution, the resolved ballots are aggregated using a standard ballot aggregation rule that is deterministic and publicly verifiable, such as majority voting.

\subsection{Information Regimes in Delegation Formation}

The information regime of delegation formation determines what a voter can observe about the partially formed delegation graph during the delegation phase. This may include, for example, delegate popularity, visible delegation chains, or aggregate statistics.

While many intermediate regimes are possible, this work focuses on two canonical regimes:

\begin{itemize}
    \item \textbf{Transparent delegation regime.} Delegation relationships are publicly observable as they form. Voters can observe delegate popularity, which may create social influence, herding, and verifiable coercion.
    \item \textbf{Sealed delegation regime.} Delegation relationships remain cryptographically hidden until a designated reveal time $T_r$, for instance the moment the votes are counted. During delegation formation, voters rely on their own beliefs about competence and trust, and cannot observe the delegation decisions of others.
\end{itemize}

A recent survey classifies liquid democracy models by features such as how many decision stages a model has, how many delegates a voter may determine, and how cycles are handled~\cite{grossi2026modeling}. None of these features captures what voters can see about the delegations of others while choices are made. The concept of information regime fills exactly this gap. In the terminology of the survey, our model is multi-stage, multi-proxy, deterministic, substitutive with a fallback ballot, constrained to an eligible delegate pool, non-strategic, static, and cycles-intolerant.

\section{System Design}
\label{sec:design}

The proposed system addresses both vulnerabilities identified in Section~\ref{sec:related} at once. As shown in Figure~\ref{fig:workflow}, it separates delegation formation during the voting phase from delegation resolution during the reveal phase, enforcing secrecy during the former and transparency during the latter.

\begin{figure}[h!]
\centering
\includegraphics[width=\textwidth]{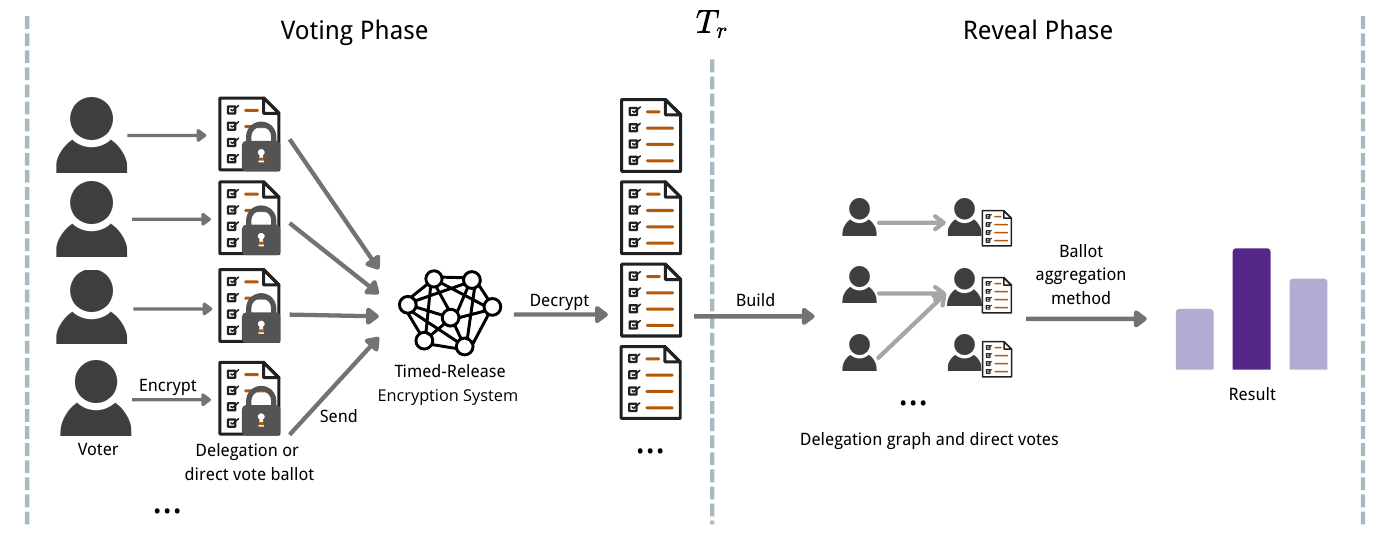}
\caption{System workflow of the proposed sealed liquid democracy mechanism. In the voting phase, voters encrypt and submit their policy tuples to the ledger; in the reveal phase, all encrypted ballots are decrypted and the system computes the final outcome.}
\label{fig:workflow}
\end{figure}

\subsection{Public Ledger}
\label{sec:ledger}

All communication in the protocol goes through a public ledger, such as a blockchain, that anyone can read. The set of eligible voter addresses is fixed and published on the ledger before voting begins. Before $T_s$, each voter posts their encrypted policy, which also carries the encrypted key shares for the reveal committee. After $T_r$, committee members post their shares, so anyone can reconstruct the keys, decrypt every policy, and recompute the resolution and the tally. The protocol relies on four properties of the ledger. (i) Entries are append-only, so no posted entry can be changed or removed. (ii) Entries are totally ordered, which decides which submission of a voter is the latest. (iii) Every entry is public, which makes the tally auditable. (iv) Block height or block timestamp serves as the clock that enforces $T_s$ and $T_r$. Before $T_r$, an observer therefore sees which addresses posted, when, and how often, but not what they posted.

\subsection{Cryptographic Primitive: Timed-Release Encryption}
To enforce sealed delegation without relying on a single trusted authority, our system utilizes a decentralized Timed-Release Encryption scheme. Traditional schemes often rely on computationally intensive time-lock puzzles, timed commitments~\cite{rivest1996timelock, boneh2000timed}, or on centralized time servers. We instead adopt a blockchain-based ``time machine'' approach~\cite{tle}, which uses a decentralized committee and verifiable secret sharing. This primitive has recently been applied to secure direct voting systems where the voters themselves act as the committee~\cite{collective-secure-voting}. Our work abstracts this layer, requiring only that the committee remains decentralized and secure.

In this scheme, a conditional-reveal committee is established via a smart contract. To encrypt a message for a future reveal time $T_r$, a voter generates a symmetric encryption key, encrypts the ballot with it, and uses Shamir's Secret Sharing~\cite{shamir1979secret} to split the key into shares. These shares are encrypted and distributed to the committee members. The smart contract ensures that committee members are only incentivized to reveal their shares once the block height or timestamp of the blockchain surpasses $T_r$. Once $T_r$ is reached, any observer can reconstruct the symmetric key from a threshold of revealed shares, allowing all sealed ballots to be decrypted. Crucially, as long as the attacker controls fewer committee members than the reconstruction threshold of before $T_r$, the pre-reveal secrecy of the ciphertexts holds. 

\subsection{Voting Phase: Sealed Delegation Formation}

Each voter constructs a policy tuple
\[
\Pi_i = (D_i, b_i),
\]
where $D_i = [d_i^{(1)}, d_i^{(2)}, \dots, d_i^{(\ell_i)}]$ is an ordered list of $\ell_i \geq 0$ ranked delegates, and $b_i$ is a direct ballot to be used if all delegations fail. A voter who votes directly submits an empty list, $\ell_i = 0$.

Instead of publishing delegation relationships directly, each voter encrypts $\Pi_i$ using the Timed-Release Encryption scheme described above, parameterized by the reveal time $T_r$. The resulting ciphertext is posted to a public ledger before the submission deadline. For a voter $i$, the encrypted policy is written as:
\[
C_i = \mathcal{E}_{T_r}(\Pi_i),
\]

\noindent where $\mathcal{E}_{T_r}$ denotes encryption under the timed-release scheme with reveal time $T_r$. Every encryption draws a fresh symmetric key and fresh randomness for the key shares, so two encryptions of the same policy look unrelated. Before encryption, every policy tuple is padded to a fixed maximum encoding size, with a bounded delegate list and a fallback slot that is always present. The length of a ciphertext therefore reveals nothing about how many delegates a voter ranked or whether they delegated at all. If policy updates are allowed before the deadline, the voter may resubmit, and only the latest valid submission is accepted. Latest is defined by the total order of the ledger. A coercer therefore cannot dictate a sequence number that would invalidate later overrides.

Delegation decisions are hidden only while voters are still making them. After $T_r$ all accepted policies are revealed, so the tally is fully auditable and the protection against influence is temporal (Appendix~\ref{app:residual}).

\subsection{Reveal Phase: Deterministic Delegation Resolution}
\label{sec:reveal}

Once policy tuples are revealed, delegation is resolved (with the algorithm in Section~\ref{subsec:algorithm}) and the ballots are aggregated deterministically. A ciphertext that decrypts to a malformed policy is treated as an abstention; under plain timed-release encryption such invalidity is only discoverable at $T_r$, since no one can inspect ciphertext contents earlier.

Voting power propagates along ranked delegation lists. For each voter $i$, delegates are evaluated in order. If the highest-ranked delegate casts a valid vote, the voting power of $i$ is assigned to that delegate. If that delegate abstains or becomes inactive, the system proceeds to the next ranked delegate. If all delegates fail, the direct ballot $b_i$ is used if provided. Only if no valid delegate and no direct ballot exist is the vote marked as abstention.

Delegation cycles are broken by advancing the voter with the smallest index in the cycle to their next ranked delegate, as Step~5 of Algorithm~\ref{alg:sealed-delegation} specifies. The specific choice does not affect correctness, and fixing it keeps resolution deterministic and publicly recomputable. A fallback ballot is invoked only once a voter has exhausted all ranked delegates.

\subsection{Algorithm for Delegation Formation and Resolution}
\label{subsec:algorithm}

Algorithm~\ref{alg:sealed-delegation} specifies the sealed delegation procedure. It separates submission, reveal, cycle resolution, fallback use, and tallying steps.

\begin{algorithm}
\caption{Sealed ranked delegation with reveal-time resolution}
\label{alg:sealed-delegation}
\begin{enumerate}
\item \textbf{Input:} voter set $N$, submission deadline $T_s$ and reveal time $T_r$ with $T_s < T_r$. For each voter $i \in N$ a policy tuple
\[
\Pi_i = (D_i, b_i),
\]
where $D_i = [d_i^{(1)}, d_i^{(2)}, \dots, d_i^{(\ell_i)}]$ is an ordered list of $\ell_i \geq 0$ ranked delegates and $b_i$ is an optional fallback ballot.

\item \textbf{Submission phase:} Before $T_s$, each voter encrypts $\Pi_i$ as
\[
C_i = \mathcal{E}_{T_r}(\Pi_i),
\]
and posts $C_i$ to the public ledger. If multiple submissions are allowed, the system retains only the latest valid submission of each voter in the total order of the ledger.

\item \textbf{Reveal phase:} At time $T_r$, the system decrypts all accepted ciphertexts and reconstructs the set of revealed policies $\{\Pi_i\}_{i \in N}$.

\item \textbf{Initialize choices:} For each voter $i \in N$, set their active choice to their highest-ranked valid delegate $d_i^{(1)}$ who is not inactive and has a revealed policy. If no such delegate exists, set their active choice to their fallback ballot $b_i$.

\item \textbf{Iterative Cycle Resolution:} Construct the directed delegation graph based on all current active choices. While cycles exist in the graph, do the following for each cycle:
\begin{itemize}
\item Select the voter with the smallest index in the cycle.
\item Advance the active choice of this voter to their next ranked delegate who has a valid revealed policy.
\item If this voter has exhausted all their ranked delegates, set their active choice to their fallback ballot. If they have no fallback ballot, mark them as abstaining.
\end{itemize}

\item \textbf{Final Resolution:} Once the graph is cycle-free, the path of every voter deterministically terminates at a direct voter, a fallback ballot, or an abstention. Route the voting power along these paths.

\item \textbf{Tally:} Aggregate all resolved ballots using the public deterministic ballot aggregation rule $F$. The final tally is publicly recomputable from the revealed policies and the resolution rule.
\end{enumerate}
\end{algorithm}

Let $n = |N|$ and let $L = \sum_{i \in N} \ell_i$ be the total length of all ranked lists. Each voter has at most one active choice, so one iteration of Step~5 finds every cycle in $O(n)$ time. Each iteration also advances at least one voter by one entry of their list or to their fallback ballot. Step~5 therefore runs at most $L + n$ iterations and takes $O(n(n + L))$ time, while Steps~4 and~6 take $O(n + L)$ time. Every list is padded to a bounded length, so $L = O(n)$ and resolution takes $O(n^2)$ time in the worst case. Step~3 reconstructs one key from $t$ shares per voter, so decryption grows linearly in $n$ for a fixed committee. The cost of Step~7 is the cost of the aggregation rule $F$.

\subsection{Security and Robustness Guarantees}
\label{sec:security}

Three guarantees follow from the algorithm. (i) After $T_r$, all accepted policies are published and resolution is deterministic, so any observer can recompute the delegation graph and reproduce the tally. (ii) Resolution terminates for all revealed policy profiles, since ranked delegation is a finite search through the delegate list of each voter and cycles are broken by the fixed rule of Step~5. (iii) Ranked delegates and fallback ballots reroute voting power when delegates abstain or fail to submit, which single-delegate systems~\cite{liquidfeedback, liquid-friesland, tally} cannot do. Section~\ref{sec: failures} quantifies the effect.

\paragraph{Threat model.} \label{sec: threatmodel} The attacker reads the entire public ledger, including the time at which every ciphertext was posted. It controls any number of voters, together with their encryption randomness. It controls fewer than $t$ of the $\kappa$ reveal committee members before $T_r$, so no coalition it assembles can open a ciphertext early. The ledger is append-only and totally ordered. The two guarantees below hold for every voter that the attacker does not control. They rely on a single assumption about the encryption, namely that before $T_r$ a sealed ciphertext reveals nothing about its content, even to an attacker holding up to $t-1$ committee shares. This is the property the decentralized timed-release construction we build on is designed to provide~\cite{tle}, and Definition~\ref{def:indtrcpa} in Appendix~\ref{app:proofs} states this formally.

\paragraph{Pre-reveal secrecy.} \label{sec: secrecy} Before $T_r$, no attacker within the threat model learns anything about how an honest voter delegated or voted. Theorem~\ref{thm:secrecy} in Appendix~\ref{app:proofs} proves this. Our design prevents the two attacks that need to observe delegations while they are still being made. Herding needs visible delegation counts, and a vote buyer who pays during formation needs contemporaneous evidence that the purchased delegation was cast. Deployments must additionally reject copied or modified ciphertexts, as discussed in Appendix~\ref{app:proofs}.

\paragraph{Receipt-freeness via resubmission.} \label{sec: receiptfree} A coercer can order a voter to encrypt a dictated policy and to hand over the ciphertext together with the randomness that produced it, as a receipt proving compliance. The coerced voter posts the dictated ciphertext, hands over the receipt, and then posts a fresh encryption of their own policy before $T_s$. Only the latest submission counts, so their own policy replaces the dictated one in the tally. Re-voting systems count only the last ballot for the same reason~\cite{springall2014estonian, juels2005coercion}. The receipt proves only what the dictated ciphertext contains. By pre-reveal secrecy, the coercer cannot tell whether the later submission repeats the dictated policy or replaces it (Theorem~\ref{thm:receiptfree}, Appendix~\ref{app:proofs}). The guarantee fails against a coercer who demands that the voter post exactly once, which Appendix~\ref{app:residual} discusses together with the remaining coercion channels.

\paragraph{Off-chain signaling.} Sealing removes delegation counts from the ledger, but voters can still communicate offline. A delegate can campaign for delegations and a voter can announce whom they delegated to, but before $T_r$ neither claim can be verified, since a delegate cannot prove how many delegations they hold and a voter who shows a receipt can still resubmit. The popularity signal that herding needs therefore degrades from an observed count to unverifiable claims, as in a conventional secret ballot election with campaigning and polls. Such claims can still sway some voters and leave a weak form of herding. Our experiments model herding by a strength $h$, which is $0$ under sealed formation and positive under transparent formation (Section~\ref{par: her_exp_slc}), so weak herding corresponds to a small $h$. As $h$ grows, vote loss under targeted delegate failures rises, but the ranked design with fallback ballots keeps it at about $3\%$ for every $h$ (Section~\ref{sec: regime_failures}).

\section{Experimental Setup and Modeling}
\label{sec:experiments}

Evaluating the behavioral impact of sealed versus transparent delegation requires observing how voters form delegation networks under different information regimes. Because no large-scale sealed delegation system is currently deployed, we evaluate our system using a data-driven experimental framework calibrated on real-world voting and demographic data. This approach allows us to isolate the effects of visibility and herding, while keeping the underlying voter preferences empirically grounded. Unless stated otherwise, reported results are means over repeated random delegate draws on a fixed electorate, 30 for the failure analyses and 50 for the regime comparisons, with $95\%$ confidence intervals. 

The experiments address three questions, each answered in one subsection of Section~\ref{sec:evaluation}. The first is what happens to the voting results when delegates fail, varying the ballot design, the failure model, and the fraction of failing delegates (Section~\ref{sec: failures}). The second is how the information regime shapes the delegation network and the outcome, varying the herding strength (Section~\ref{sec: power_concentration}). The third is how a delegate should vote for their constituency and when delegation adds accuracy, varying the representation scheme and the delegation rate (Section~\ref{sec: rep_style}).

Prominent models of political participation tie abstention to the cost of deciding~\cite{downs1957economic}. Majority abstention is the empirical norm in the settings we target. Average turnout in Swiss federal popular elections was $47\%$ over 2011 to 2024 and even lower before~\cite{bfs2025stimmbeteiligung, serdult2014referendums}. Participatory budgeting elections typically draw one to three percent of eligible residents~\cite{doherty2024who}, and the Aarau campaign itself drew $1703$ voters from roughly $14{,}000$ enfranchised residents~\cite{pournaras2025upgrading}. Abstention is also not random. Education is among the most consistent predictors of turnout, and abstainers differ systematically from voters~\cite{smets2013embarrassment, lijphart1997unequal, citrin2003what}. Referencing these turnouts, we conservatively model the 50\% of voters with lower expertise as abstainers.

Whether the abstainers would choose to delegate when given the option is less settled. Proxy voting in Dutch national elections is used by about one in ten voters~\cite{odihr2025netherlands}. However, no study shows that offering a delegation option leads voters who would otherwise abstain to delegate in a real election. We therefore treat this as a design premise and report how much accuracy delegation recovers as a function of the delegation rate in Section~\ref{sec: delegation_model}.

\subsection{Datasets}
\label{sec: datasets}

Most of our data comes from participatory budgeting, a now widely adopted democratic process in which residents vote to allocate part of a public budget across proposed projects~\cite{cabannes2004participatory}. The design and computational analysis of participatory budgeting voting rules is an active research area~\cite{aziz2021participatory, peters2021proportional}. Our evaluation is grounded in the Aarau election below; two further sources are used to check that the findings generalize beyond a single electorate.
\begin{itemize}
    \item \textbf{Aarau (Stadtidee)}~\cite{pournaras2025upgrading}: a municipal participatory budgeting election in which 1{,}700 citizens cast valid ballots directly, without delegation, on local projects under a fixed budget. Of these voters, 222 also completed demographic and political surveys. This is our primary dataset, as it uniquely pairs real ballots with a survey. The survey is rich enough to measure the expertise of each voter and to model how voters choose delegates who resemble them in age, district, and political self-placement.
    \item \textbf{Pabulib (20 datasets)}~\cite{pabulib}: twenty real-world participatory budgeting elections from Warsaw, Amsterdam, and Budapest, with 6{,}000 to 96{,}000 voters each, all analyzed on their full electorates. These provide ballots and project costs, and the Warsaw instances include voter age and, for some, neighborhood. Lacking a survey, we proxy competence by voter engagement, that is, the number of projects a voter approved.
    \item \textbf{CES 2022 (US elections)}~\cite{ces2022}: the Cooperative Election Study common content, a nationally representative survey of 60{,}000 US adults, all of whom we analyze. We use it to test the expertise hypothesis with an objective competence measure, a political knowledge battery asking which party controls the US House and Senate, together with rich demographics. Participants express binary support/oppose positions on 34 policy items spanning guns, abortion, health care, and immigration, and the outcome is the per-item majority.
\end{itemize}
The delegation model below is described for Aarau and applied to each dataset using the voter features it provides.

\subsection{Data and Model Calibration}
\label{sec: calibration}

Voters in the Aarau election did not delegate. They voted directly, and we use their direct ballots as the ground truth preferences. The goal of liquid democracy is to approximate the collective will of the electorate. The outcome of this undelegated population therefore serves as the baseline against which we measure both the representational loss caused by abstention and the recovery provided by each delegation regime.

Among the electorate, 222 participants completed additional demographic and political surveys. We use these survey responses to construct a synthetic expertise proxy for each voter, detailed in Appendix~\ref{app: expertise}. We divide these 222 participants by expertise, the top 10\% as high-expertise, the next 40\% as mid-expertise, and the bottom 50\% as low-expertise. In the experiments below, the low-expertise group acts as the delegators.

Before introducing abstention or delegation, we verify that the 222 survey respondents are broadly representative of the full electorate. Under the Method of Equal Shares, aggregating their ballots reproduces $89\%$ of the pairwise comparisons in the outcome of all 1{,}700 valid ballots, so the subsample is close to the full population. We then consider a counterfactual in which the 50\% of respondents with the lowest expertise abstain. Outcome accuracy, defined in Section~\ref{sec: metric_rules}, then drops to $0.847$ under the Method of Equal Shares and to $0.811$ under the greedy rule. Under the Method of Equal Shares, abstention changes four of the sixteen funded projects. It drops a free clothes exchange and a boule project in the Telli district, which receive $64\%$ of their points from the abstaining half, and funds two projects that receive only $26\%$ of their points from it. The abstaining half is also younger, with a mean age of $43$ against $52$ to $54$ for the other respondents. When one subgroup is systematically less likely to participate, the projects it supports therefore drop out of the winning bundle. The resulting accuracy is the baseline that delegation must improve upon.

\paragraph{Independent delegate selection.} \label{par: ind_exp_slc} In our model, delegators choose delegates who are easier for them to trust or recognize. Platforms such as LiquidFeedback offer a searchable delegate directory with self-disclosed demographic information, neighborhood affiliations, and political leanings. We assume voters use it to select delegates who are socially, geographically, or politically relatable. We operationalize this using similarity in age, district, and political self-placement. The detailed weights and procedure used for the Stadtidee dataset are given in Appendix~\ref{app: ind_exp_sel}.

\paragraph{Delegate selection with herding.} \label{par: her_exp_slc} In the transparent setting, delegators can additionally observe which delegates are already attracting support. We model this by increasing the probability of choosing a delegate who is already popular. This captures the reinforcement effects driven by visibility documented in Section~\ref{sec:related} and in models of information cascades~\cite{bikhchandani1992theory, banerjee1992simple, barabasi1999emergence, salganik2008leading}. The functional form follows the preferential delegation model of G\"olz et al., where $h$ plays the role of the in-degree exponent~\cite{golz2018fluid}. No measurement of $h$ exists for delegation networks. Attachment exponents measured in comparable online networks fall between $0$ and $1.5$, and we treat this as the plausible range for $h$, as detailed in Appendix~\ref{app: seq_del_her}. Operationally, the sealed regime corresponds to independent draws from the similarity weights, with $h=0$, and the transparent regime to sequential draws with increasing $h>0$. To measure how much the outcome depends on the arrival order, we rerun the election with $40$ independent random orders. 

\paragraph{Representative ballots.} \label{par: rep_ballots} In the Stadtidee voting system, each voter distributes a fixed budget of points across the projects, so a constituency is summarized by how many points its members assign to each project. A representative delegate, as defined in Section~\ref{par: rep_del_mod}, then casts a ballot formed in one of three ways. The mean rule assigns each project the mean number of points across the constituency. The median rule assigns the median number of points. The election rule holds an internal election in the constituency, which elects a bundle of $k$ projects by sequential proportional approval voting over the approval sets of its members. Here $k$ is the median number of projects a member supports, and the delegate casts an approval ballot backing that bundle. We evaluate all three, since the choice turns out to materially affect the result (Section~\ref{sec: rep_style}).

\paragraph{Delegation rate.} The delegation rate $\rho$ is the fraction of abstainers who delegate. The abstainers who delegate at a lower rate are a subset of those who delegate at a higher rate, so results at different rates differ only in how many abstainers delegate.

\paragraph{Transitive model.} In the Aarau model, delegates do not delegate further. To study chains of delegations, as deployed platforms allow, we add an abstract agent-based model in which voters may delegate to a delegate who has also delegated. Voters arrive one after another, as in the Aarau model, and under transparency each voter sees how many delegations every member has received from earlier arrivals. Voting power is then measured along delegation chains, for example A $\to$ B $\to$ C.

\subsection{Metrics and Aggregation Rules}
\label{sec: metric_rules}

\paragraph{Accuracy metric.} We treat the outcome produced when every simulated voter casts their own ballot as the ground truth benchmark, computed under the same aggregation rule as the outcome it is compared with. We measure outcome accuracy as the pairwise agreement between an outcome and the benchmark, using a measure based on Condorcet pairwise comparisons~\cite{majumdar2024generative}. For every unordered pair of projects $(a,b)$, we determine which project wins in pairwise comparison under both the benchmark and the outcome of interest; accuracy is the fraction of project pairs on which the two outcomes agree. An accuracy of $1.0$ denotes perfect agreement with full participation. To explain where accuracy is lost, we also count the projects that differ from the benchmark. A project is wrongly funded if it wins in an outcome but not under full participation, and wrongly dropped in the reverse case.

\paragraph{Concentration.} We measure how concentrated delegated power is with the Gini coefficient of voting power among delegates and with the share of first-choice delegations held by the largest delegate. Section~\ref{sec: sensitivity} also reports the Gini coefficient across all voters, including those who vote directly. In the transitive model we also report the maximum transitive voting share, the largest share of all voting power that reaches a single delegate through delegation chains, and the effective number of delegates. The effective number of delegates is one divided by the sum of the squared shares of voting power. If several delegates hold equal shares, it equals their number, and it approaches $1$ when one delegate holds almost all the power.

\paragraph{Aggregation rules.} The winning bundle depends on how ballots are aggregated under the budget constraint. We evaluate two rules with the fixed budget of CHF~50{,}000. The first is the Method of Equal Shares, a proportional rule with strong fairness guarantees for participatory budgeting~\cite{peters2021proportional}. In a pre-registered behavioral study of Z\"urich participatory budgeting, citizens also perceived it as fairer than the conventional greedy alternative~\cite{yang2024designing}. Aarau used this rule, with the points of each voter as their utilities and the Add1 completion of the official implementation~\cite{equalshares_aarau}. On all 1{,}700 valid ballots, our implementation selects 17 projects for CHF~49{,}400, and their costs match those of the 17 official winners one for one~\cite{aarau2025evaluation}. The second is the utilitarian greedy rule currently in use in participatory budgeting~\cite{pabulib}, which funds projects in order of their total points until the budget runs out. It funds only seven or eight mostly expensive projects, so a small change in the ballots can move a large share of the budget. Under the greedy rule the 222 survey respondents reproduce only $66\%$ of the outcome of all valid ballots, which is why both rules are compared against the outcome of the simulated voters themselves. Figures report the Method of Equal Shares unless stated otherwise, and the text reports the greedy rule alongside.

\subsection{Delegate Failure Model}
\label{sec: failure_model}

Delegation concentrates the responsibility for casting many ballots onto a small set of delegates, so the reliability of the outcome depends on those delegates actually voting. We therefore model delegate failure explicitly. A failed delegate submits no ballot before the tallying deadline, for example due to inactivity, abstention, technical failure, or a denial-of-service attack, and consequently both their own vote and the votes routed to them are at risk. Such failures are common in practice; across eighteen DAOs, delegates vote on only about a third of proposals on average~\cite{hall2024what}, and Section~\ref{sec:intro} documents targeted interference with the popular delegates in deployed systems. We study two failure modes:
\begin{enumerate}
    \item \textbf{Uniform random failures:} a fraction $p_{\mathrm{f}}$ of delegates, chosen uniformly at random, fail. This captures benign unreliability.
    \item \textbf{Targeted failures:} the $p_{\mathrm{f}}$ fraction of delegates with the highest delegated voting power fail. This captures an attacker who targets the most consequential hubs, or correlated failure of the popular delegates.
\end{enumerate}

Against these we compare three designs. In the single-delegate baseline, each delegator names only one delegate. In the ranked design, each delegator names a first delegate and adds a further one with probability $q$, up to three delegates. Section~\ref{sec: failures} sweeps over the value of $q$. Ranked delegation with a personal fallback ballot additionally lets a voter whose ranked delegates all fail revert to their own direct ballot. Because not every voter can be expected to prepare a full direct ballot, we treat the fraction $\phi$ of delegators who provide a fallback ballot as an experimental parameter. We also vary the information regime under which the delegation graph itself is formed. The same failures are then applied to a graph formed under sealed delegation and to a graph formed under transparent delegation. Under the transparent regime the graph is formed by sequential herded selection at $h \in \{0.3, 1.0, 1.5\}$, with the ranked list of each voter sampled at their arrival turn. Targeted failures then remove the delegates who received the most delegations in each formed graph. Paired random seeds compare the regimes on identical populations, over $30$ repetitions on Aarau and over $5$ seeds on five Pabulib elections subsampled to $2{,}500$ voters. Unless stated otherwise, delegations are formed under the sealed regime with the personal representation scheme, and the delegation graph stays fixed while the ballot design varies. Failure pertains only to whether a delegate submits a valid ballot. Resolution and tallying follow the reveal-phase protocol of Section~\ref{sec:reveal} and are unaffected. All failure experiments reuse the same random draws. The same delegate choices and the same failure orderings are used at every failure fraction and for every ballot design. The comparison between conditions is therefore free of sampling noise.

\section{Experimental Evaluation}
\label{sec:evaluation}

Section~\ref{sec:related} identifies two vulnerabilities of transparent liquid democracy, namely vote loss once delegated power concentrates in a few delegates who may fail, and influence over delegation choices while they are formed. The evaluation tests whether the proposed mechanism addresses both, along the three questions of Section~\ref{sec:experiments}. Throughout, accuracy is the agreement metric of Section~\ref{sec: metric_rules}, under the Method of Equal Shares unless the greedy rule is named. Vote loss and concentration do not depend on the aggregation rule.

\subsection{Resilience to Delegate Failure} \label{sec: failures}

Without failures, the default configuration of Section~\ref{sec: failure_model} reaches an accuracy of $0.831$, and $0.710$ under the greedy rule.

\begin{figure}[h!]
\centering
\includegraphics[width=\textwidth]{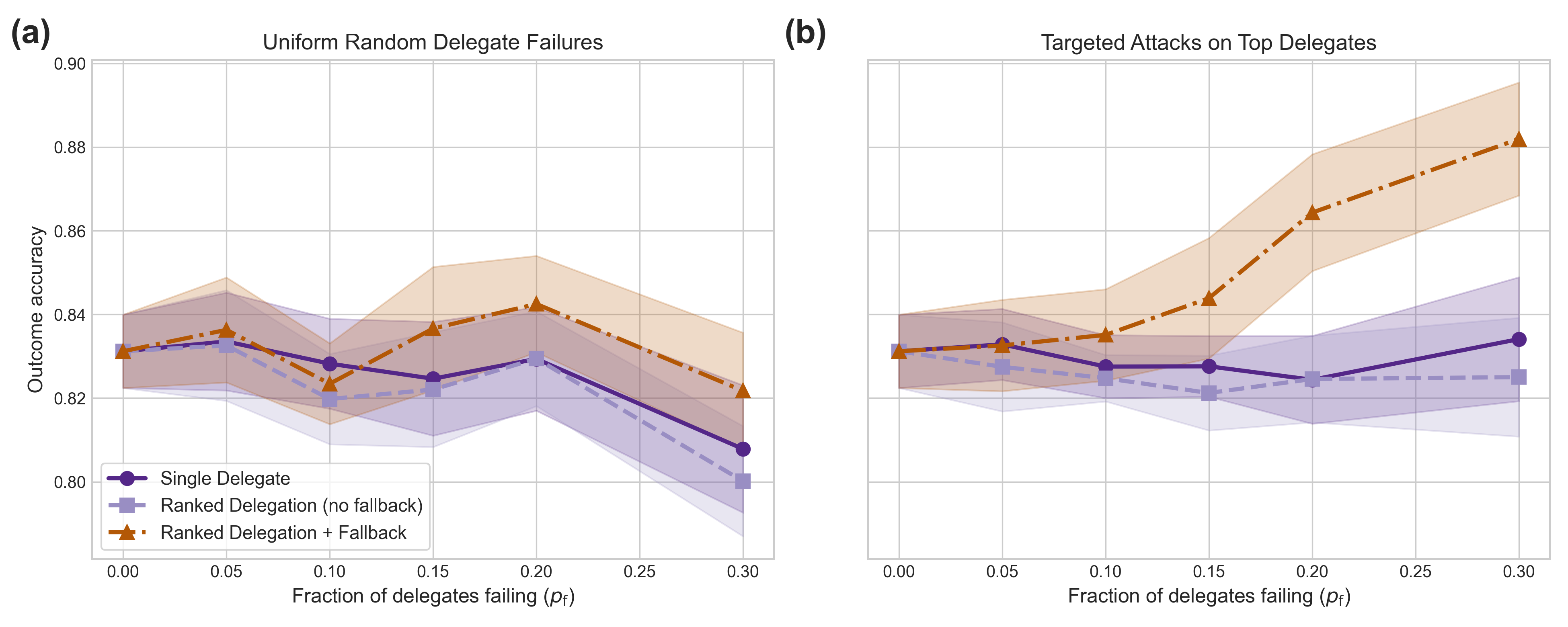}
\caption{Outcome accuracy under the Method of Equal Shares as a function of the fraction of failing delegates, for uniform random failures (left) and targeted attacks on the delegates who hold the most delegations (right). Shaded bands are $95\%$ confidence intervals over 30 repetitions. Under targeted attacks, ranked delegation with a personal fallback ballot is the most accurate design, because stranded voters revert to their own ballots.}
\label{fig:failure_accuracy}
\end{figure}

\paragraph{Outcome accuracy.} Figure~\ref{fig:failure_accuracy} reports outcome accuracy as $p_{\mathrm{f}}$ grows. Under uniform random failures, accuracy is stable for all three designs, staying within roughly $0.80$ to $0.84$. Losing a few random delegates does not shift the outcome in any particular direction, because the remaining ballots still resemble the full electorate. Under targeted attacks the designs diverge. At $p_{\mathrm{f}}=0.3$ ranked delegation with a personal fallback ballot reaches $0.882 \pm 0.014$, $6\%$ above the single-delegate design at $0.834 \pm 0.015$ and the highest at every failure level. The greedy rule gives the same order, with $0.812$ for the fallback design against $0.791$ for single delegation. The fallback ballot reduces the wrongly dropped projects from $1.9$ to $1.1$ per run and the wrongly funded ones from $2.1$ to $1.5$. The free clothes exchange, which the delegators mostly support, is dropped in $87\%$ of single-delegate runs and in $50\%$ of runs with fallback.

\paragraph{Vote loss.} The designs differ more in how many ballots go uncounted (Table~\ref{tab:loss}). Under targeted attacks at $p_{\mathrm{f}}=0.3$, the single-delegate design loses $30.2\%$ of ballots, ranked delegation still loses $20.8\%$, and ranked delegation with fallback loses only $2.7\%$; under uniform failures the single-delegate loss is $16.3\%$. An accuracy score alone can therefore hide a large number of uncounted ballots. Only the combination of ranked delegation and a personal fallback ballot keeps vote loss low and accuracy high under targeted failures, which are the failures that power concentration makes likely. 

Ranked delegation helps only as much as voters use it. Each delegator adds a further delegate with probability $q$, up to three. Under targeted failures at $p_{\mathrm{f}}=0.3$, raising $q$ from $0$ to $1$ lowers the vote loss of ranked delegation without a fallback from $30.2\%$, the loss of a single delegate, to $7.7\%$. With a fallback ballot the loss is $2.7\%$ for every $q$, so it no longer depends on how many delegates voters rank.

\begin{table}[t]
\centering
\caption{Vote loss (fraction of the electorate whose ballot is not counted) at $p_{\mathrm{f}}=0.3$, mean over 30 repetitions. Ranked delegation with a personal fallback ballot is the only design whose loss stays low under targeted attacks.}
\label{tab:loss}
\begin{tabular}{lcc}
\hline
Design & Uniform failures & Targeted failures \\
\hline
Single delegate & 16.3\% & 30.2\% \\
Ranked delegation & 10.3\% & 20.8\% \\
Ranked delegation + fallback ballot & 2.7\% & 2.7\% \\
\hline
\end{tabular}
\end{table}

\paragraph{The $3\%$ floor and fallback provision.} When every delegator supplies a fallback ballot, vote loss no longer depends on the shape of the network. Every voter whose delegates all fail reverts to their own ballot, so the only unresolved votes are the ballots of the failed delegates themselves. The loss therefore equals $p_{\mathrm{f}}$ times the delegate share of the electorate, about $3\%$ at $p_{\mathrm{f}}=0.3$, on any electorate and under any aggregation rule. This value is a floor, since no design can recover the ballots of the failed delegates themselves. All twenty Pabulib datasets on their full electorates show exactly this loss, as Figure~\ref{fig:failures_multi} shows. The loss of a single-delegate system falls from $26.4\%$ to $3.0\%$ under targeted attacks, and from $17.6\%$ to $3.0\%$ under uniform failures. The agreement across elections follows from this arithmetic. Figure~\ref{fig:fallback_provision} sweeps the fraction $\phi$ of delegators who provide fallback ballots. Vote loss under targeted attacks falls linearly from $17.1\%$ at $\phi=0$, where the design reduces to ranked delegation, to the $3.0\%$ floor at $\phi=1$ on Pabulib, and from $20.6\%$ to $2.7\%$ on Aarau, where outcome accuracy also rises by $6\%$, from $0.826$ to $0.873$, and by $1\%$ under the greedy rule. Ranked delegation on its own therefore recovers only about a third of the votes that single delegation loses, and the fallback ballot recovers the rest. Section~\ref{sec: regime_failures} shows that the loss without a fallback grows further once the delegation graph is concentrated.

\begin{figure}[h!]
\centering
\includegraphics[width=\textwidth]{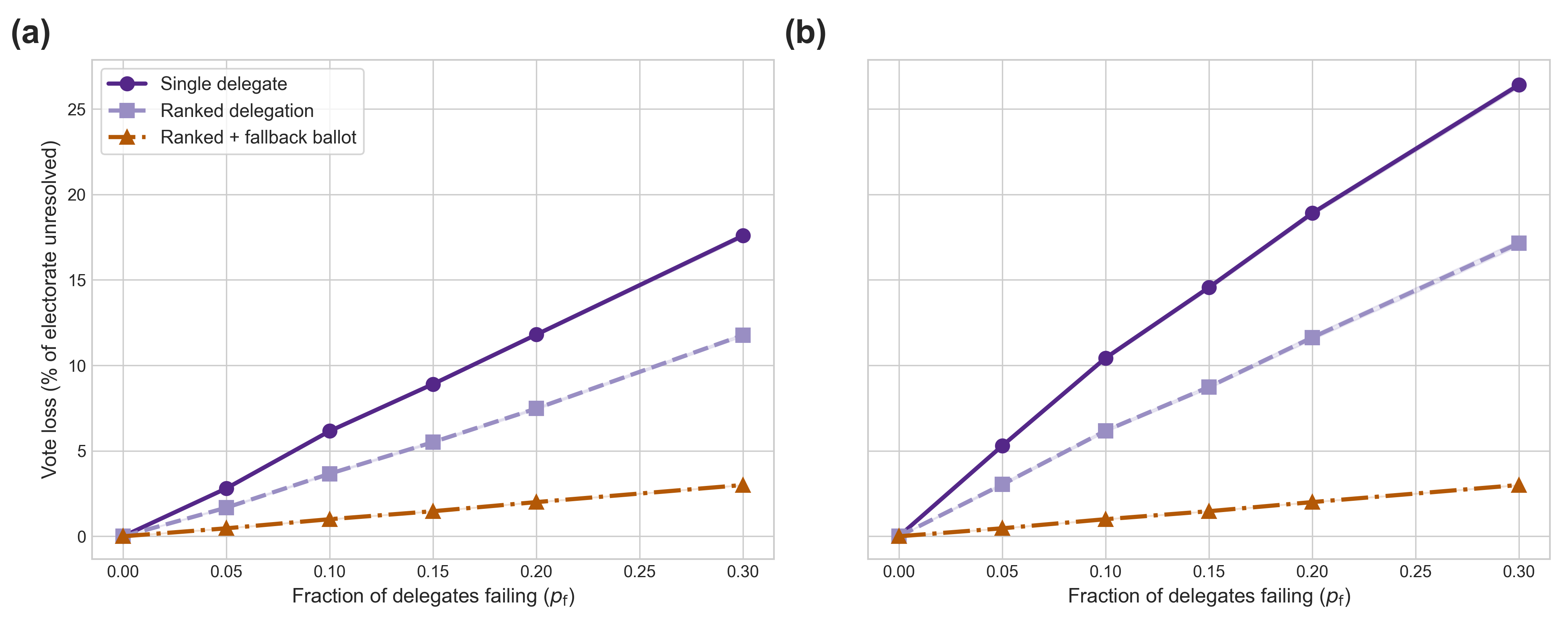}
\caption{Vote loss under (a) uniform random delegate failures and (b) targeted failures of the delegates who hold the most delegations. Values are averages over 20 real participatory budgeting elections from Pabulib on their full electorates, with $95\%$ confidence intervals across elections. Ranked delegation with a personal fallback ballot keeps vote loss near $3\%$ even when $30\%$ of delegates are targeted, versus $26\%$ for a single-delegate design. The loss is the same in all twenty elections.}
\label{fig:failures_multi}
\end{figure}

\begin{figure}[h!]
\centering
\includegraphics[width=\textwidth]{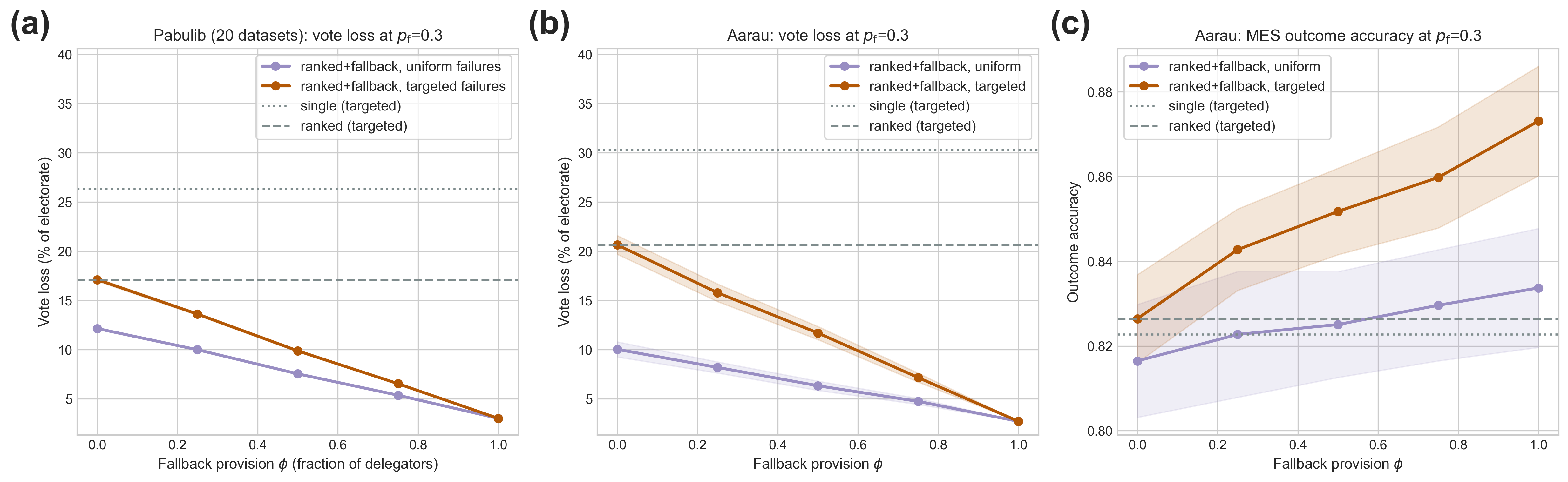}
\caption{Sensitivity of the ranked design with fallback to the fraction $\phi$ of delegators who actually supply a fallback ballot, at $p_{\mathrm{f}}=0.3$. Vote loss falls linearly in $\phi$ toward the floor of $p_{\mathrm{f}}$ times the delegate share, on the twenty Pabulib elections (left) and on Aarau (middle). The Aarau outcome accuracy rises with $\phi$ (right). The single-delegate design and the ranked design without fallback are shown as reference lines.}
\label{fig:fallback_provision}
\end{figure}

\subsubsection{Failures on Graphs Formed under Each Regime} \label{sec: regime_failures}

The failure analysis above holds the delegation graph fixed and varies the ballot design. Here we instead form the graph under each information regime, as described in Section~\ref{sec: failure_model}, and test whether concentration leads to higher vote loss under failure.

Stronger herding concentrates the delegations in fewer delegates (Figure~\ref{fig:regime_failures}). On Aarau the share of first-choice delegations held by the largest delegate rises from $0.14$ under sealing to $0.64$ at $h=1.5$, with the Gini coefficient of voting power among delegates rising from $0.34$ to $0.73$. Section~\ref{sec: power_concentration} explains how visibility produces this concentration. This concentration raises vote loss under targeted failure. Under targeted failures at $p_{\mathrm{f}}=0.3$ with single delegation, vote loss rises from $30.2\%$ on sealed graphs to $50.6\%$ at $h=1.5$ on Aarau, and from $26.2\%$ to $48.9\%$ on Pabulib. Ranked delegation without a fallback shows the same regime dependence, rising from $20.7\%$ to $41.5\%$ on Aarau and from $16.8\%$ to $43.2\%$ on Pabulib, so under strong herding ranking recovers little of what single delegation loses. Outcome accuracy is a poor measure here, because after targeted failures it mostly reflects how many stranded voters fall back on their own ballots, which favors the most concentrated graphs. Vote loss isolates the effect of the regime. Only the ranked design with fallback is unaffected by the regime. Its loss stays at the floor of about $3\%$ under every regime, because every voter whose delegates all fail reverts to their own ballot, no matter how concentrated the graph is. Where fallback provision is incomplete, sealed formation therefore keeps the loss at about half of its value under strong herding.

\begin{figure}[h!]
\centering
\includegraphics[width=\textwidth]{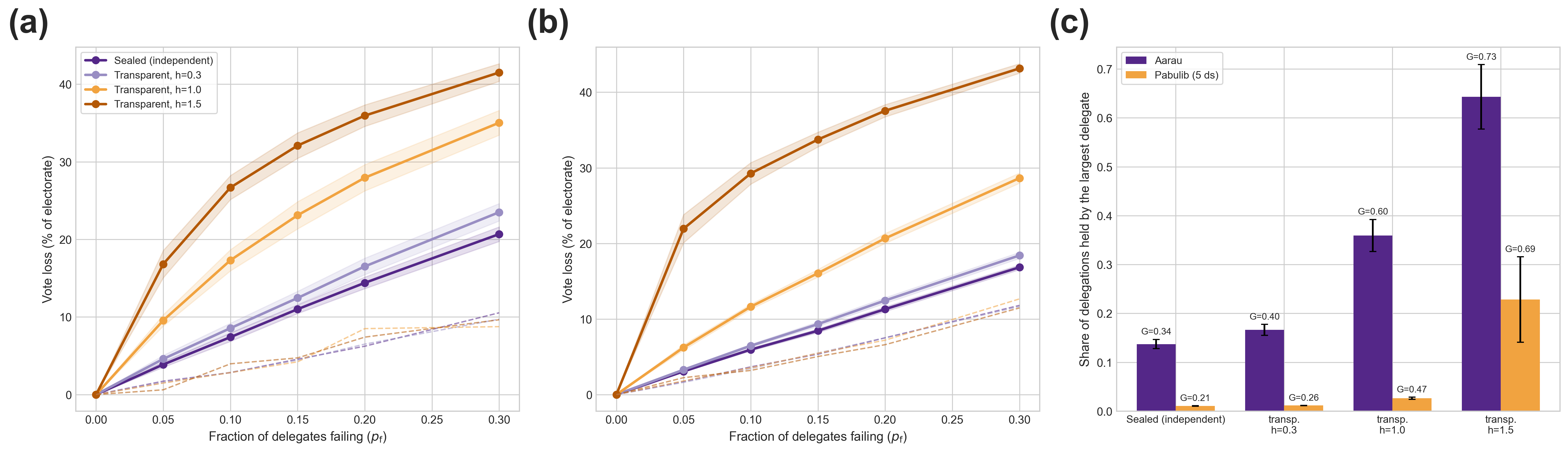}
\caption{Delegate failures on graphs formed under sealed versus transparent delegation, where targeted failures remove the top delegates of each realized graph. (a) Aarau vote loss for the ranked design, with solid lines for targeted and dashed lines for uniform failures, over $30$ paired repetitions with $95\%$ confidence intervals. (b) Vote loss on five Pabulib elections. (c) Realized concentration of the formed graphs, measured as the share of first-choice delegations held by the largest delegate, with the Gini coefficient of voting power among delegates annotated. Transparent formation concentrates delegations in fewer delegates, and concentration raises the loss from targeted failures. The ranked design with fallback stays at the $3\%$ floor under every regime.}
\label{fig:regime_failures}
\end{figure}

\subsection{Concentration and Order Dependence under Transparency} \label{sec: power_concentration}

Section~\ref{sec: failures} shows that graphs formed under transparent delegation lose more votes when their most popular delegates fail. This subsection measures how transparency concentrates delegated power, and shows that the concentration also makes the outcome depend on the order in which voters act. On the Aarau electorate, where low expertise voters delegate to a fixed pool of experts and delegates do not delegate further, transparency concentrates power among the delegates, as Section~\ref{sec: failures} shows, but changes the distribution of voting power across all voters only modestly. Counting every voter, including those who vote directly, the Gini coefficient rises by about $10\%$, from $0.44$ with no herding to $0.48$ under strong herding at $h=1.5$, as shown in Section~\ref{sec: sensitivity}.

Concentration becomes severe once delegation is transitive, meaning that delegates themselves delegate to more popular delegates, as deployed platforms allow. In the transitive model of Section~\ref{sec: calibration} the transparent setting ends with almost all voting power in one delegate. This is the known behaviour of strongly superlinear preferential attachment, under which a single node comes to connect to nearly all others~\cite{krapivsky2000connectivity}. Over $15$ runs of the model, the effective number of delegates falls to approximately $1$, and the maximum transitive voting share of a single delegate approaches $1.0$. Nearly all delegation chains funnel into one ``super delegate'' (Figure~\ref{fig:transparent_graph}a). Under sealing, by contrast, power remains spread across roughly $10$ effective delegates with a maximum share near $0.24$ (Figure~\ref{fig:transparent_graph}b). This matches the emergence of super voters in LiquidFeedback and the concentration of delegated power in DAO governance documented in Section~\ref{sec:related}~\cite{kling2015voting, fritsch2024analyzing}.

\begin{figure}[h!]
\centering
\includegraphics[width=\textwidth]{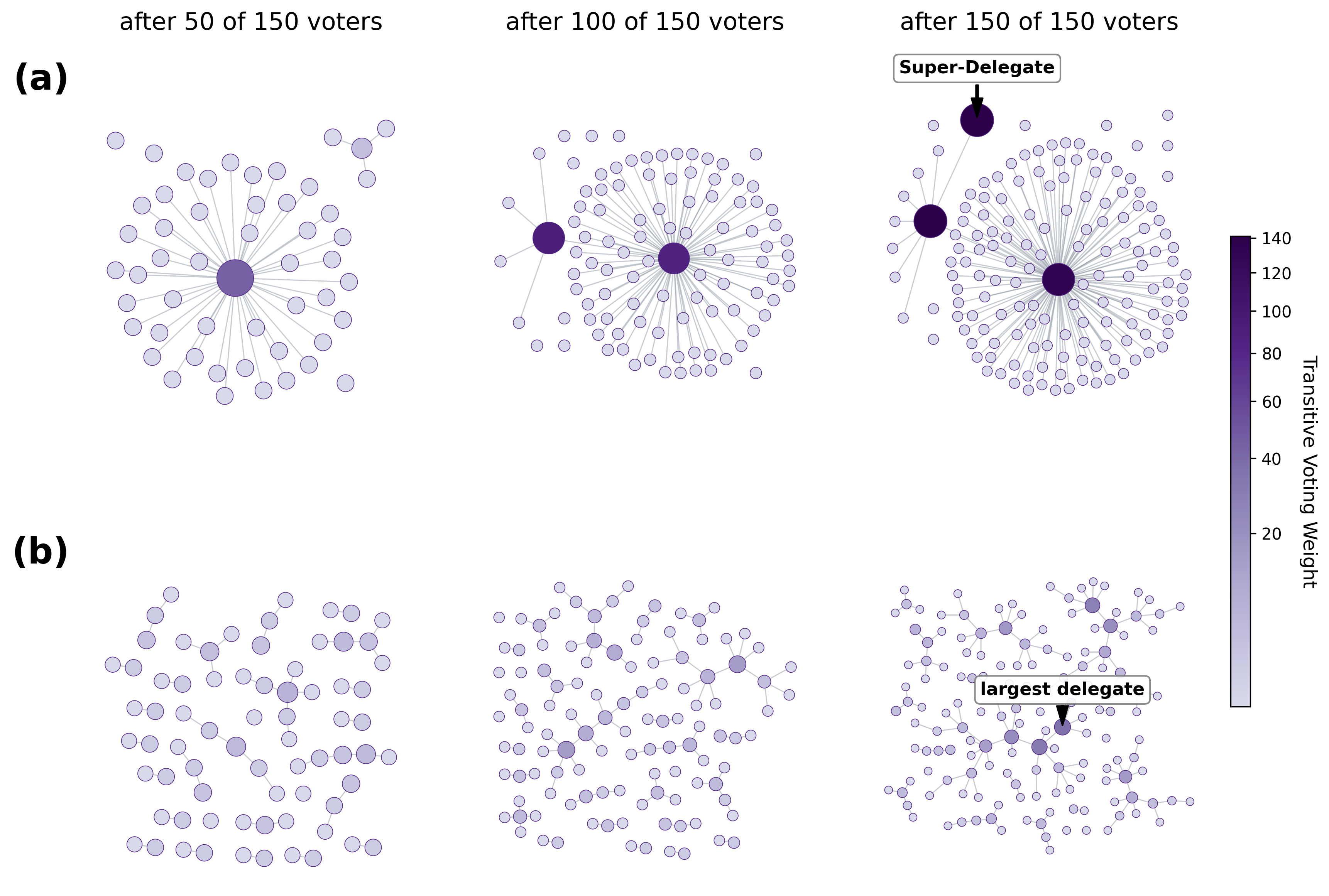}
\caption{Delegation graphs in the transitive model after 50, 100, and all 150 voters have arrived, under (a) transparent and (b) sealed formation. Node size and shade encode transitive voting weight, from light for low to dark for high. Under transparency one delegate collects most delegations from the first arrivals onward, while under sealing voting power stays spread across roughly $10$ effective delegates.}
\label{fig:transparent_graph}
\end{figure}

\subsubsection{Outcome Stability under Arrival Order} \label{sec: stability}

Because transparent delegates are chosen sequentially and popularity reinforces popularity, which delegate ends up with the most delegations depends on the order in which voters happen to act. The outcome therefore depends partly on that order. Figure~\ref{fig:pathdep} reports mean accuracy with its spread across runs as the herding exponent $h$ grows. Under sealed delegation the choices do not depend on order, so the outcome is stable throughout, with a run-to-run accuracy standard deviation of $0.03$ at every $h$. Since $h$ has no effect when delegations are hidden, the spread under sealing comes only from the random choice of delegates. Under transparency the spread doubles to $0.06$ at $h=3$. The dissimilarity between winning bundles of different runs triples, from $0.11$ to $0.32$. Mean accuracy itself falls by $9\%$, from $0.83$ to $0.76$, because under strong herding most delegations end at whichever delegate gained an early lead. At $h=3$ this delegate holds $94\%$ of the first-choice delegations, against $14\%$ under sealing, and $78\%$ of the wrongly funded projects are on their ballot, against $22\%$ under sealing. These projects are also more expensive, CHF~10{,}400 on average against CHF~6{,}300, so they crowd other projects out of the budget, and the wrongly dropped projects rise from $2.1$ to $3.6$ per run. The greedy rule is less stable under every regime, with a dissimilarity of $0.40$ already under sealing, and transparency raises it to $0.59$. Under transparency, greedy accuracy falls by $13\%$ as $h$ grows from $0$ to $3$. At $h=0$ the two regimes differ only by sampling noise, which is large under the greedy rule. Visibility during formation therefore makes the outcome less accurate and less predictable under both rules.

\begin{figure}[t]
\centering
\includegraphics[width=\textwidth]{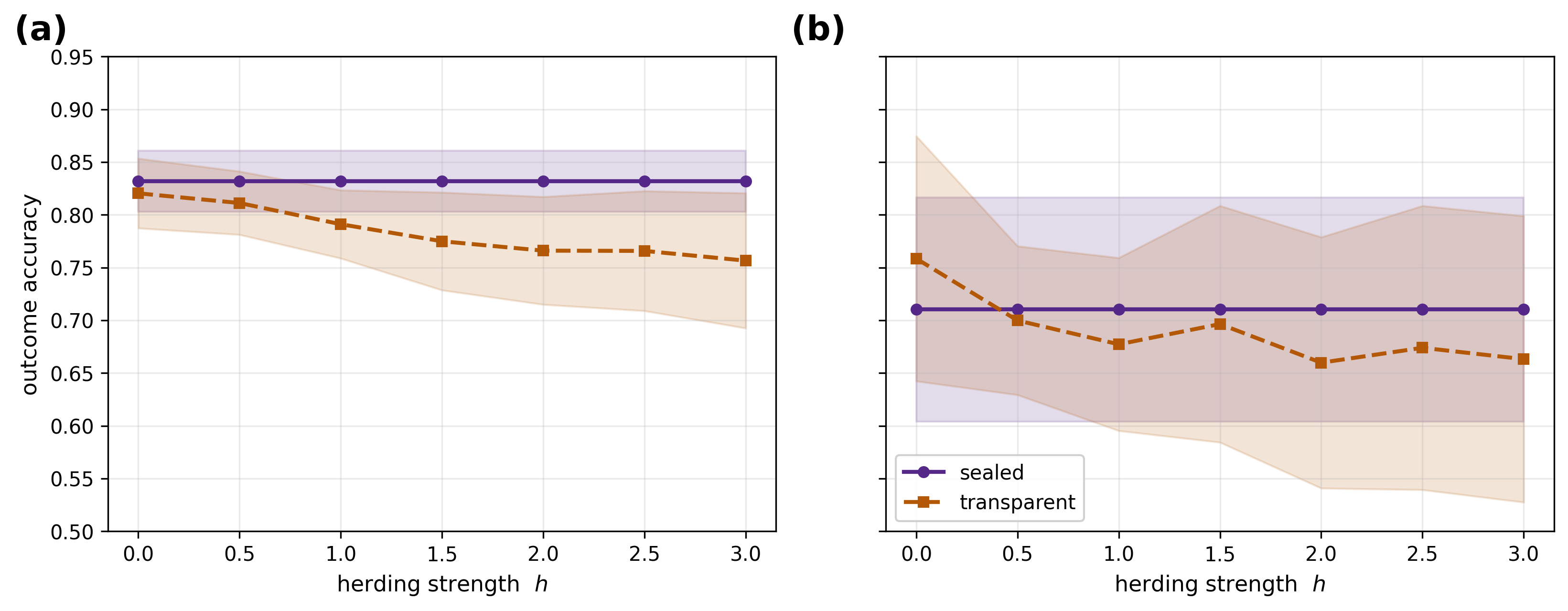}
\caption{Outcome stability under sealed versus transparent delegation as the herding strength $h$ grows, under (a) the Method of Equal Shares and (b) the greedy rule. Lines show mean outcome accuracy over 40 independent runs with random arrival orders, and shaded bands show one standard deviation across runs. Under transparency the spread widens and mean accuracy drifts down under both rules, so the outcome depends on the arrival order. The greedy rule varies more between runs under every regime.}
\label{fig:pathdep}
\end{figure}

\subsection{Representation Scheme and Outcome Accuracy} \label{sec: delegation_model} \label{sec: rep_style}

How the delegate votes determines how much accuracy delegation recovers (Figure~\ref{fig:accuracy}). When every abstainer delegates, a personal delegate lowers accuracy by $3\%$ against the abstention baseline of $0.847$, to $0.825 \pm 0.008$, and by $10\%$ under the greedy rule. Personal delegation does not bring back the projects that abstention dropped. The free clothes exchange and the boule project remain unfunded in $98\%$ and $90\%$ of runs, because the delegates are older voters who give these projects few points. Experts are $10\%$ of the respondents but supply only $4\%$ and $5\%$ of the points of these two projects. A representative delegate who casts the mean point ballot of their constituency carries the support of the delegators into the count. This raises accuracy by $7\%$ over the baseline, to $0.908 \pm 0.010$, the highest of any configuration and a significant improvement under a Wilcoxon signed-rank test with $p < 10^{-14}$. Both dropped projects return in every run, and the wrongly dropped projects fall from $2.1$ to $0.8$ per run. A representative delegate who casts the median ballot backs a project only when at least half of the constituency does, and reaches $0.850 \pm 0.017$, no better than the baseline. The internal election elects a bundle of the typical size by proportional approval voting inside the constituency and reaches $0.862 \pm 0.014$, $2\%$ above the baseline. Both peak near $0.90$ at delegation rates of $20\%$ to $40\%$ and then decline as the constituencies grow. Every representative scheme is above the baseline from a delegation rate of $10\%$ onward. 

\begin{figure}[h!]
\centering
\includegraphics[width=\textwidth]{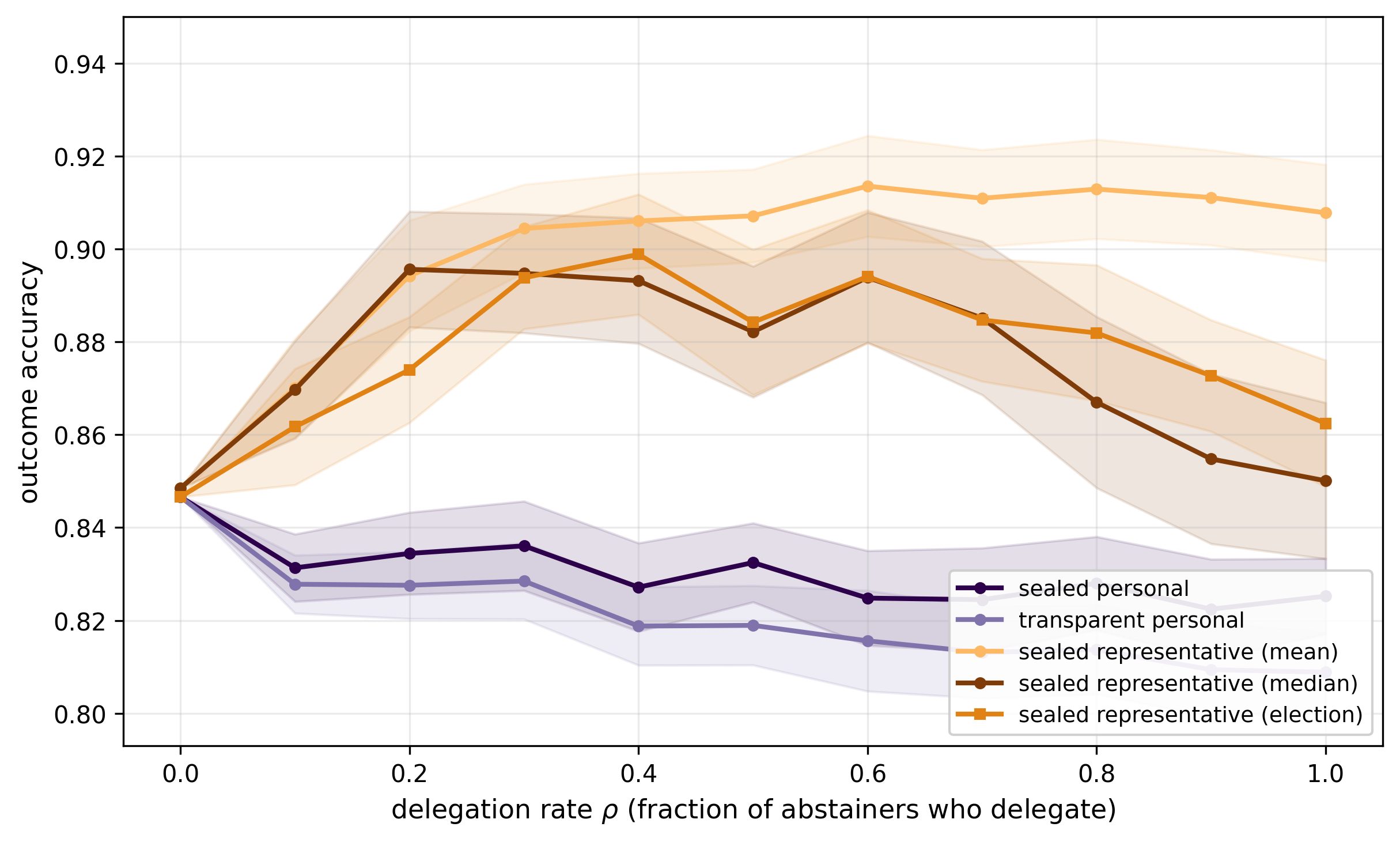}
\caption{Outcome accuracy on Aarau under each representation scheme as a function of the delegation rate $\rho$, under the Method of Equal Shares. Means over 50 repetitions with $95\%$ confidence bands. At $\rho=0$ only the 50\% of respondents with the highest expertise vote, giving the $0.847$ baseline. The mean scheme is the most accurate at every rate from $0.3$ onward.}
\label{fig:accuracy}
\end{figure}

Table~\ref{tab:crossdata} repeats the comparison on twenty Pabulib elections and on 60{,}000 US voters from the CES~2022 study, where expertise is an objective measure of political knowledge. Personal delegation changes accuracy by a median of $0\%$ on Pabulib and lowers it by $3\%$ on CES, one policy item out of $34$. Under the Method of Equal Shares, the median scheme lowers accuracy by a median of $5\%$ on Pabulib and by up to $23\%$ in the citywide Warsaw elections. Pabulib voters approve few of the projects on the ballot, so a majority of a constituency rarely agrees on a project and median ballots are sparse. In the 2024 citywide Warsaw election, median ballots approve $2.8$ projects against $8.7$ on the average ballot, which changes the winning bundle. The internal election approves a bundle of the typical size, beats the median scheme in seventeen of twenty elections, cuts the worst-case loss to $13\%$, and raises accuracy by about $10\%$ in the three elections where abstention costs the most accuracy. Under the greedy rule, the median scheme instead raises accuracy by a median of $2\%$. The safe representation scheme therefore depends on the ballot format and the aggregation rule, with the mean ballot for point ballots and the internal election for sparse approval ballots. Delegation can only recover accuracy that abstention removed, so where abstainers resemble the rest of the electorate, the value of the mechanism lies in the vote loss and concentration results of Sections~\ref{sec: failures} and~\ref{sec: power_concentration}.

\begin{table}[t]
\centering
\caption{Relative change in outcome accuracy, in percent, against the abstention baseline of each dataset, by representation scheme. Pabulib cells give the median over the 20 elections and, in brackets, the lowest and highest change. Participatory budgeting results use the Method of Equal Shares or the greedy rule on full electorates. The representative column uses the mean rule for Aarau under the Method of Equal Shares and the median rule otherwise, because under the greedy rule the mean rule is an arithmetic identity. CES uses per-item majorities on 34 policy items, so its accuracy moves in steps of $3\%$, and the internal election is defined only for approval ballots.}
\label{tab:crossdata}
\footnotesize
\setlength{\tabcolsep}{3.5pt}
\begin{tabular}{llcccc}
\hline
Dataset & Rule & Baseline & $\Delta$ personal & $\Delta$ representative & $\Delta$ election \\
\hline
Aarau (222) & Equal Shares & 0.847 & $-3$ & $+7$ & $+2$ \\
 & greedy & 0.811 & $-10$ & $-6$ & $-6$ \\
Pabulib (20) & Equal Shares & 0.75--0.97 & $0$ $[-14,+3]$ & $-5$ $[-23,+6]$ & $+0.5$ $[-13,+10]$ \\
 & greedy & 0.85--1.00 & $0$ $[-28,+7]$ & $+2$ $[-1,+14]$ & $+0.2$ $[-11,+16]$ \\
CES 2022 & majority & 0.971 & $-3$ & $0$ & -- \\
\hline
\end{tabular}
\end{table}

\subsection{Threats to Validity} \label{sec: sensitivity} \label{sec: limitations}

\paragraph{Internal validity.} Because the similarity model of delegate selection has several free parameters, we check that the conclusions do not depend on the chosen values, as shown in Figure~\ref{fig:sensitivity}. Sweeping the district bonus $\beta_g \in [1,5]$, the age scale $\sigma_a \in [5,80]$, and the political scale $\sigma_p \in [0.5,8]$ keeps sealed outcome accuracy between $0.80$ and $0.84$. Accuracy rises by up to $5\%$ as the age and political scales grow, which makes delegators less selective about similarity, and the default values sit in the middle of this range. The default values therefore do not drive the results. Increasing the herding exponent $h$ raises power concentration at every step (Figure~\ref{fig:sensitivity}a). Herding also lowers outcome accuracy, by $7\%$ between $h=0$ and $h=1.5$ under the Method of Equal Shares and by $13\%$ under the greedy rule, because most delegations then go to a few popular delegates and the tallied ballots come from fewer distinct people. The default $h=0.3$ is in the range where herding is mild. The herding mechanism itself rests on a model. Existing liquid democracy experiments do not vary whether voters can see the delegation choices of others~\cite{campbell2022liquid, berinsky2025tracking}, so the parameter $h$ remains uncalibrated (Appendix~\ref{app: seq_del_her}) and we report results across its plausible range.

\begin{figure}[h!]
\centering
\includegraphics[width=\textwidth]{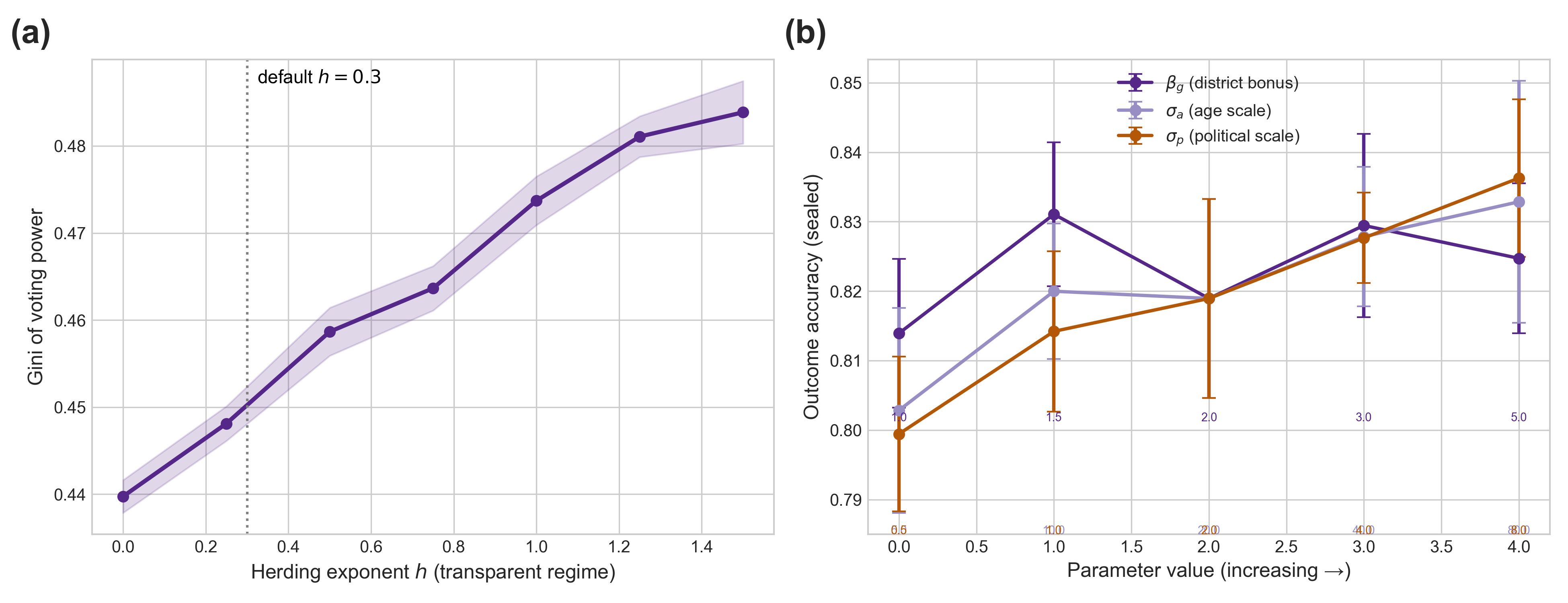}
\caption{Sensitivity to the model parameters. (a) Increasing the herding exponent $h$ raises power concentration (Gini coefficient of voting power across all voters) in the transparent regime. (b) Sealed outcome accuracy under the Method of Equal Shares varies by at most $5\%$ across the similarity parameters $\beta_g$, $\sigma_a$, and $\sigma_p$, with the default values in the middle of the range. Bands and bars are $95\%$ confidence intervals over 20 repetitions.}
\label{fig:sensitivity}
\end{figure}

\paragraph{Construct validity.} Competence is measured by a proxy in every dataset. Only the Aarau survey is rich enough to calibrate the full similarity model of delegate selection with an expertise measure taken from the survey itself, while Pabulib uses engagement and CES an objective knowledge battery.

\paragraph{External validity.} In our model of the Aarau electorate, non-experts delegate to experts and delegates do not delegate further. It therefore understates the concentration that arises when delegations chain through several delegates, as in deployed systems, which we approximate only with the transitive model of Section~\ref{sec: calibration}. Participation is modeled under a similar caveat, since whether abstainers would delegate in a real election is a design premise of our experiments (Section~\ref{sec: calibration}).

\section{Conclusion and Future Work}
\label{sec:conclusion}

We propose a liquid democracy mechanism with sealed delegation formation using timed-release cryptography, combined with ranked multi-delegation and fallback ballots for robustness under delegate failure. Experiments grounded in a municipal participatory budgeting election and checked against twenty further participatory budgeting elections and a large US survey yield a consistent picture. Sealed delegation reduces power concentration and the resulting vulnerability to targeted delegate failures, and it keeps outcomes stable under herding. Without sealing, the outcome depends on the arbitrary order in which voters act. Ranked delegation with a personal fallback ballot keeps vote loss at a floor near $3\%$ even when $30\%$ of delegates are targeted, provided that delegators supply fallback ballots. A single delegate loses $26\%$ in the same attack. When provision is incomplete or delegations have concentrated under transparent formation, vote loss rises to as much as $49\%$, and sealing the formation phase is what keeps that risk low.

More broadly, the paper argues for a specific timing of transparency. Delegation choices are hidden only while they are being formed and are disclosed for counting and verification. Protection therefore ends at the reveal, which makes the mechanism best suited to settings where post-hoc vote publicity is already accepted.

\paragraph{Future work.}

Field deployments in live civic platforms or DAOs are needed to measure the real-world behavioral impact of sealed versus transparent delegation. Integrating the mechanism with privacy-preserving decentralized identity systems would allow delegates to prove domain expertise or community standing using verifiable credentials. This would foster informed delegation while protecting voter anonymity. 

\section{Acknowledgments}
This work is supported by a UKRI Future Leaders Fellowship (UKRI3495): “\emph{Smart City Commons of Trustworthy Distributed Intelligence - ARTIO 2.0}”

\bibliography{reference}
\bibliographystyle{splncs04}

\clearpage
\appendix
\section{Notation}
\label{app:notation}

\begin{table}[H]
\centering
\small
\caption{Notation used in the paper.}
\label{tab:notation}
\vspace{5pt}
\begin{tabular}{@{}lp{8.6cm}@{}}
\toprule
Symbol & Meaning \\
\midrule
\multicolumn{2}{@{}l}{\textit{Voters, ballots, and delegation}} \\
$N$, $n$ & set of voters and its size \\
$A$, $m$ & set of policy alternatives and its size \\
$b_i$ & direct ballot of voter $i$, used as a fallback when delegations fail \\
$D_i = [d_i^{(1)}, \dots, d_i^{(\ell_i)}]$ & ranked delegate list of voter $i$, with $\ell_i \geq 0$ entries \\
$L$ & total length of all ranked delegate lists, $\sum_{i \in N} \ell_i$ \\
$\Pi_i = (D_i, b_i)$ & policy tuple of voter $i$ \\
$W(g)$ & voting weight of terminal voter $g$ \\
$F$ & public deterministic ballot aggregation rule \\
\midrule
\multicolumn{2}{@{}l}{\textit{Sealing and security}} \\
$T_s$, $T_r$ & submission deadline and reveal time, with $T_s < T_r$ \\
$\mathcal{E}_{T_r}$ & timed-release encryption with reveal time $T_r$. The proofs write $\mathcal{E}_{T_r}(\cdot\,;\,r)$ to name the randomness $r$ \\
$C_i$, $r_i$ & sealed policy of voter $i$ and the random values drawn to seal it, which are the symmetric key, the random coefficients of the sharing polynomial, and the randomness that encrypts each share \\
$\kappa$, $t$ & reveal committee size and reconstruction threshold \\
$\lambda$ & security parameter \\
$\mathcal{A}$, $\mathcal{B}$ & attacker on the voting system, and the attacker on the encryption built from it in the proofs \\
$\Pi_{\mathcal{A}}$, $C_{\mathcal{A}}$ & policy dictated by a coercer and its sealed ciphertext \\
\midrule
\multicolumn{2}{@{}l}{\textit{Delegation model and experiments}} \\
$e_i$ & expertise score of voter $i$ \\
$N_{\mathrm{d}}$, $N_{\mathrm{r}}$ & delegators and eligible delegates \\
$w_{ij}$, $\tilde{w}_{ij}(\tau)$ & selection weight of delegator $i$ for delegate $j$, without and with herding \\
$s_i$, $y_i$, $p_i$ & district, age, and left-right political position of voter $i$ \\
$\delta_{ij}$, $\alpha_{ij}$, $\pi_{ij}$ & district, age, and political proximity of $i$ and $j$ \\
$\beta_g$, $\sigma_a$, $\sigma_p$ & district bonus, age decay, and political decay parameters \\
$c_j(\tau)$ & number of delegations received by $j$ before sequential position $\tau$ \\
$h$ & herding strength, with $h = 0$ under sealed formation \\
$\rho$ & fraction of abstainers who delegate \\
$\phi$ & fraction of delegators who supply a fallback ballot \\
$p_{\mathrm{f}}$ & fraction of delegates that fail \\
$q$ & probability that a delegator adds another delegate to their ranked list \\
\bottomrule
\end{tabular}
\end{table}

\section{Security Proofs}
\label{app:proofs}

Sealed delegation protects honest voters from two kinds of pressure during the voting phase. Theorem~\ref{thm:secrecy} proves that no one learns how an honest voter delegated before $T_r$. Theorem~\ref{thm:receiptfree} proves that a coercer cannot verify that a voter obeyed. Table~\ref{tab:attacks} lists the attacks that the threat model of Section~\ref{sec: threatmodel} allows and the protection against each.

\begin{table}[H]
\centering
\footnotesize
\setlength{\tabcolsep}{4pt}
\caption{Attacks on sealed delegation and the protection against each. The two theorems of this appendix prevent the first three attacks. The next two require a measure in the deployment. Early opening is excluded by the threat model, and sealing cannot prevent the last three.}
\label{tab:attacks}
\begin{tabular}{@{}>{\raggedright\arraybackslash}p{2.0cm}>{\raggedright\arraybackslash}p{3.6cm}>{\raggedright\arraybackslash}p{3.6cm}>{\raggedright\arraybackslash}p{1.8cm}@{}}
\toprule
Attack & What the attacker does & Protection & Where \\
\midrule
Herding & Reads the ledger during voting to find the delegates with the most delegations & Prevented, because policies stay sealed until $T_r$ & Theorem~\ref{thm:secrecy} \\
Vote buying & Pays for a delegation and checks the ledger to see that it was cast & Prevented, because the ledger shows only sealed policies and a receipt can be overridden & Theorems~\ref{thm:secrecy} and~\ref{thm:receiptfree} \\
Coercion with a receipt & Dictates a policy and demands the ciphertext and its randomness as proof & Prevented, because the voter can replace the dictated policy unnoticed & Theorem~\ref{thm:receiptfree} \\
Ballot copying & Reposts the sealed policy of another voter under its own pseudonym & Prevented if submissions prove knowledge of their content and duplicates are rejected & Appendix~\ref{app:secrecy} \\
Counting submissions & Demands that the voter post exactly once & Prevented if every client resubmits on a fixed schedule & Appendix~\ref{app:receipt} \\
Early opening & Corrupts $t$ or more committee members to open policies before $T_r$ & Excluded by the threat model & Section~\ref{sec: threatmodel} \\
Watching the voter & Checks what the voter posts last, demands no submission, or posts with a stolen credential & Not covered, requires anonymous credentials & Appendix~\ref{app:residual} \\
Waiting for the reveal & Reads the policy of the voter after it is revealed at $T_r$ & Not covered, because sealing protects only the voting phase & Appendix~\ref{app:residual} \\
Timing & Links the time of a resubmission to a visit to the voter & Not modelled & Appendix~\ref{app:residual} \\
\bottomrule
\end{tabular}
\end{table}

\subsection{Proof Approach}
\label{app:approach}

No one can currently prove outright that an encryption scheme cannot be broken. Security proofs in cryptography therefore assume one property of a well studied building block and derive the other guarantees from it. Here the building block is the timed-release encryption of Section~\ref{sec:design}. The assumed property is that a sealed message reveals nothing about its content before $T_r$. The proofs use a reduction. A reduction shows that any attacker who breaks the voting system can be turned into an attacker who breaks the encryption. Since the encryption is assumed secure, no attacker can break the voting system.

The attacker is any program that may make random choices and whose running time grows at most polynomially in a security parameter $\lambda$. Cryptographers call such an attacker probabilistic polynomial-time. The security parameter sets how hard the encryption is to break, and the key length in bits is a typical choice. Larger values of $\lambda$ make every attack harder.

Secrecy is measured with a guessing game. The attacker picks two messages, receives the encryption of one of them selected by a fair coin, and guesses which one it received. An attacker who guesses at random is right half of the time. The advantage of an attacker is the distance between its probability of guessing correctly and one half. An advantage of $0$ means that the ciphertext reveals nothing, and an advantage of $1/2$ means that the attacker always guesses correctly.

An advantage is negligible in $\lambda$ if it is smaller than $1/p(\lambda)$ for every polynomial $p$ once $\lambda$ is large enough. In words, it shrinks faster than one divided by any polynomial as $\lambda$ grows. The value $2^{-\lambda}$ is an example. At $\lambda = 128$ it equals the chance of guessing a key of $128$ bits at the first try, so a negligible advantage is of no use in practice.

\subsection{The Sealing Assumption}

Both guarantees rest on one property of the timed-release encryption. Before $T_r$, a sealed message reveals nothing about its content, even to an attacker who holds the key shares of up to $t-1$ committee members. Definition~\ref{def:indtrcpa} states this property as the guessing game of Appendix~\ref{app:approach}. Cryptographers call this kind of property indistinguishability under chosen-plaintext attack, the standard notion of security for encryption. Indistinguishability means that the attacker cannot tell the encryptions of two messages apart. A chosen-plaintext attack is one in which the attacker picks these messages itself. Our version adds the reveal time and the corrupted committee members.

\begin{definition}[Sealing security]
\label{def:indtrcpa}
The following game between a challenger and an attacker $\mathcal{A}$ runs before time $T_r$.
\begin{enumerate}
    \item The challenger sets up the encryption scheme and the committee, and gives $\mathcal{A}$ the public parameters.
    \item At any point of the game, $\mathcal{A}$ may corrupt up to $t-1$ committee members. For each corrupted member, $\mathcal{A}$ receives the full internal state of that member, including every key share the member holds.
    \item $\mathcal{A}$ chooses two messages $M_0$ and $M_1$ of equal length.
    \item The challenger flips a fair coin $b \in \{0,1\}$ and returns the challenge ciphertext $c^* = \mathcal{E}_{T_r}(M_b)$, sealed with fresh randomness. It distributes the key shares of $c^*$ to the committee exactly as in the real scheme.
    \item $\mathcal{A}$ outputs a guess $b'$ and wins if $b' = b$.
\end{enumerate}
The timed-release encryption is sealing secure if every probabilistic polynomial-time attacker has a negligible advantage $\bigl|\Pr[b'=b] - \tfrac{1}{2}\bigr|$ in this game.
\end{definition}

The bound of $t-1$ corrupted members matches the threat model. The construction we build on is designed for this bound, since its secret sharing recovers a key only from $t$ or more shares~\cite{tle}. The two messages must have equal length because a ciphertext reveals the length of its content. Every policy is therefore padded to a fixed size before encryption (Section~\ref{sec:design}). Sealing security is the only cryptographic assumption of the proofs.

\subsection{Pre-Reveal Secrecy}
\label{app:secrecy}

Herding and vote buying both need to see delegations while voters are still making them (Section~\ref{sec: secrecy}). Theorem~\ref{thm:secrecy} shows that the ledger reveals nothing about the policy of an honest voter before $T_r$. The attacker may even know two candidate policies of the voter, for example a delegation to one delegate or to another, and still cannot tell which one the voter submitted.

\begin{theorem}[Pre-reveal secrecy]
\label{thm:secrecy}
Suppose the timed-release encryption is sealing secure (Definition~\ref{def:indtrcpa}). Fix an honest voter $i$ and any two padded policy tuples $\Pi_0$ and $\Pi_1$ that $i$ might submit. Under the threat model of Section~\ref{sec: threatmodel}, no probabilistic polynomial-time attacker can tell before $T_r$ whether $i$ submitted $\Pi_0$ or $\Pi_1$, except with negligible advantage.
\end{theorem}

\paragraph{Intuition.} Suppose an attacker $\mathcal{A}$ could tell whether voter $i$ sealed $\Pi_0$ or $\Pi_1$. We use $\mathcal{A}$ to build a second attacker $\mathcal{B}$ that wins the guessing game of Definition~\ref{def:indtrcpa} (Figure~\ref{fig:reduction}). The attacker $\mathcal{B}$ plays two roles at once. Towards the challenger it acts as an attacker in the guessing game. Towards $\mathcal{A}$ it acts as the whole election, with the ledger, the voters, and the committee. $\mathcal{B}$ inserts the challenge ciphertext into this simulated election as the submission of voter $i$. The challenge ciphertext seals $\Pi_0$ or $\Pi_1$ depending on the coin of the challenger. A correct guess by $\mathcal{A}$ about voter $i$ is therefore a correct guess about the coin, and $\mathcal{B}$ passes it on. If $\mathcal{A}$ guessed well, $\mathcal{B}$ would win the guessing game, which sealing security rules out.

\begin{figure}[t]
\centering
\includegraphics[width=\textwidth]{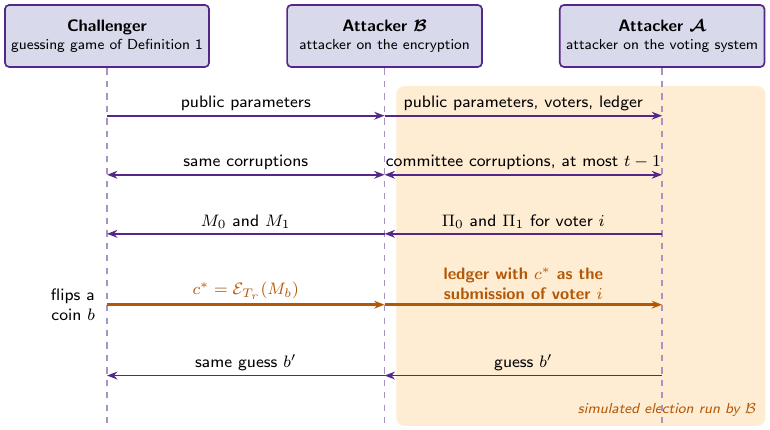}
\caption{The construction in the proof of Theorem~\ref{thm:secrecy}, with time running downwards. The attacker $\mathcal{B}$ plays the guessing game of Definition~\ref{def:indtrcpa} against the challenger on the left and runs a simulated election for the attacker $\mathcal{A}$ on the right. It inserts the challenge ciphertext $c^*$ as the submission of voter $i$, shown in orange, and passes on the guess of $\mathcal{A}$. Whenever $\mathcal{A}$ tells the two policies of voter $i$ apart, $\mathcal{B}$ wins the guessing game.}
\label{fig:reduction}
\end{figure}

\begin{proof}
Let $\epsilon$ be the advantage of $\mathcal{A}$ in telling $\Pi_0$ from $\Pi_1$. The attacker $\mathcal{B}$ plays the game of Definition~\ref{def:indtrcpa} and at the same time runs $\mathcal{A}$ inside a simulated election.
\begin{enumerate}
    \item \textbf{Setup.} $\mathcal{B}$ receives the public parameters from its challenger. It creates the ledger and every voter as the real protocol would, and starts $\mathcal{A}$.
    \item \textbf{Voters.} Whenever $\mathcal{A}$ corrupts a voter or posts for a voter it controls, $\mathcal{B}$ responds as the real protocol would. It can do so because it created these voters and knows their secrets. Voter $i$ is never corrupted, as the threat model requires.
    \item \textbf{Committee members.} Whenever $\mathcal{A}$ corrupts a committee member, $\mathcal{B}$ corrupts the same member in its own game and passes the state of that member to $\mathcal{A}$. $\mathcal{B}$ adds the key shares it created itself when it sealed the policies of the other voters. The bound of $t-1$ corruptions carries over unchanged.
    \item \textbf{Challenge.} When $\mathcal{A}$ names the two candidate policies $\Pi_0$ and $\Pi_1$ of voter $i$, $\mathcal{B}$ passes their padded encodings to its challenger as $M_0$ and $M_1$. It receives $c^* = \mathcal{E}_{T_r}(\Pi_b)$ and posts $c^*$ on the simulated ledger as the submission of voter $i$.
    \item \textbf{Guess.} $\mathcal{B}$ runs $\mathcal{A}$ until $T_r$ and outputs the guess $b'$ of $\mathcal{A}$ as its own.
\end{enumerate}
It remains to check that $\mathcal{A}$ cannot notice that it runs inside a simulation. The coin $b$ affects only two things that $\mathcal{A}$ sees, the challenge ciphertext $c^*$ and its key shares. Both come from the challenger and are formed exactly as in the real scheme. $\mathcal{B}$ creates everything else exactly as the real protocol would. What $\mathcal{A}$ sees in the simulation is therefore identical to what it would see in a real election in which voter $i$ submits $\Pi_b$. $\mathcal{A}$ thus guesses $b$ correctly in the simulation exactly as often as in a real election, and $\mathcal{B}$ wins exactly when $\mathcal{A}$ does. The advantage of $\mathcal{B}$ in the guessing game therefore equals $\epsilon$. Sealing security says that this advantage is negligible, so $\epsilon$ is negligible.
\end{proof}

\paragraph{Ballot copying.} Sealing security says nothing about an attacker who copies or modifies a ciphertext. Such an attacker could take the sealed policy of a target voter and post a copy or a modified version under a pseudonym it controls, an attack documented for the Helios voting system~\cite{cortier2013attacking}. Deployments should therefore reject duplicate ciphertexts at submission time. They should also require each submission to carry a proof that the submitter knows its content, or use a variant of the encryption that rejects modified ciphertexts. Theorem~\ref{thm:secrecy} assumes that every submission is formed independently.

\subsection{Receipt-Freeness via Resubmission}
\label{app:receipt}

A coercer or a vote buyer wants proof that a voter obeyed. Under sealed delegation the natural proof is a receipt. The coercer dictates a policy $\Pi_{\mathcal{A}}$ and demands the sealed ciphertext together with the randomness used to seal it. With both in hand, the coercer can repeat the encryption and check that the posted ciphertext seals $\Pi_{\mathcal{A}}$. Theorem~\ref{thm:receiptfree} shows that such a receipt says nothing about the policy that is tallied, because the voter can replace the dictated policy unnoticed.

\begin{theorem}[Resubmission receipt-freeness]
\label{thm:receiptfree}
Suppose the ledger is append-only, the tally uses only the latest valid submission of each voter in the total order of the ledger, and the coerced voter can post one further submission before the deadline $T_s$. Let a coercer hold a receipt $(C_{\mathcal{A}}, r_i)$ for a dictated policy $\Pi_{\mathcal{A}}$ from voter $i$. Then $i$ can have their own policy $\Pi_i$ tallied instead. Before $T_r$, the coercer cannot tell $i$ apart from an obedient voter who posts a fresh encryption of $\Pi_{\mathcal{A}}$, except with negligible advantage.
\end{theorem}

\begin{proof}
The voter proceeds in two steps.
\begin{enumerate}
    \item \textbf{Satisfy the coercer.} Encrypt the dictated policy with randomness $r_i$ as $C_{\mathcal{A}} = \mathcal{E}_{T_r}(\Pi_{\mathcal{A}}; r_i)$, post it, and hand over the pair $(C_{\mathcal{A}}, r_i)$. The coercer recomputes the encryption of $\Pi_{\mathcal{A}}$ under $r_i$, sees that it matches the posted ciphertext, and accepts the receipt.
    \item \textbf{Override with their own policy.} Encrypt their own policy with fresh randomness $r'_i$ as $C_i = \mathcal{E}_{T_r}(\Pi_i; r'_i)$ and post it before $T_s$. The tally uses the latest valid submission in ledger order, so the ballot tallied at $T_r$ is $C_i$ and $C_{\mathcal{A}}$ is superseded.
\end{enumerate}
It remains to show that the coercer cannot detect the override before $T_r$. Compare the defecting voter with an obedient voter who also posts a second ciphertext, a fresh encryption of $\Pi_{\mathcal{A}}$ under randomness $r''_i$. In both cases the coercer sees the receipt pair $(C_{\mathcal{A}}, r_i)$ followed by one further ciphertext from voter $i$. The surrendered randomness $r_i$ certifies only $C_{\mathcal{A}}$. It says nothing about the later ciphertext, which was formed with fresh randomness the coercer never sees. The later ciphertext seals either $\Pi_{\mathcal{A}}$ or $\Pi_i$, both padded to the same length. By Theorem~\ref{thm:secrecy}, no probabilistic polynomial-time coercer can tell these two cases apart before $T_r$ except with negligible advantage. The receipt therefore proves nothing about the policy that is tallied.
\end{proof}

Table~\ref{tab:receipt} lists what the coercer sees in the two cases.

\begin{table}[t]
\centering
\small
\caption{What the coercer sees before $T_r$ when a coerced voter obeys and when they defect. The two columns agree on everything visible before $T_r$, so the receipt cannot tell the cases apart.}
\label{tab:receipt}
\begin{tabular}{p{0.16\textwidth}p{0.34\textwidth}p{0.34\textwidth}}
\toprule
Event & Obedient voter & Defecting voter \\
\midrule
First submission & $C_{\mathcal{A}}$, sealing $\Pi_{\mathcal{A}}$ with $r_i$ & $C_{\mathcal{A}}$, sealing $\Pi_{\mathcal{A}}$ with $r_i$ \\
Receipt handed over & $(C_{\mathcal{A}}, r_i)$ & $(C_{\mathcal{A}}, r_i)$ \\
Second submission & fresh sealing of $\Pi_{\mathcal{A}}$ & fresh sealing of $\Pi_i$ \\
Visible before $T_r$ & $(C_{\mathcal{A}}, r_i)$ and one opaque ciphertext & $(C_{\mathcal{A}}, r_i)$ and one opaque ciphertext \\
Tallied at $T_r$ & $\Pi_{\mathcal{A}}$ & $\Pi_i$ \\
\bottomrule
\end{tabular}
\end{table}

\paragraph{Coercers who count submissions.} Theorem~\ref{thm:receiptfree} holds against a coercer who tolerates a second submission or cannot rule one out. It fails against a coercer who demands that the voter post exactly once. Submissions are attributable to a pseudonym, so a second ciphertext is then visible defection. A deployment can close this channel by having every client post a fresh encryption of its current policy on a fixed schedule. Submission counts are then identical across voters by construction, and a resubmission carries no information.

\subsection{Residual Attacks}
\label{app:residual}

Three further coercion attacks lie inside the threat model and are not covered by Theorem~\ref{thm:receiptfree}. A coercer can demand that the dictated ciphertext be posted just before the deadline and check that nothing follows it. A coercer can demand that no submission appear at all, which the public ledger makes checkable. A coercer who obtains the signing credential of a voter can submit in their place. All three afflict any scheme in which submissions are attributable to a voter, and countering them requires unlinkable submissions through anonymous credentials in the style of Juels, Catalano, and Jakobsson~\cite{juels2005coercion}. Submission timing is also visible and can leak, for instance when a voter resubmits shortly after a visit from a coercer. We do not model this side channel. After $T_r$, every accepted policy is revealed together with the pseudonym that posted it. A coercer who waits for the reveal therefore learns whether the voter obeyed. Sealing protects against pressure that needs evidence during the voting phase, such as a payment before $T_r$. Finally, the published delegation graph of one round identifies the delegates most relied upon in the next, which are the targets of the failure analysis in Section~\ref{sec: regime_failures}.

\section{Stadtidee Dataset Delegation Modeling}
\label{app:stadtidee}

We construct an expertise score for each of the 222 survey participants in Stadtidee~\cite{pournaras2025upgrading} based on self-reported political efficacy and interaction frequency with municipal actors. This score is only a proxy used in our experiments. It does not claim to measure expertise objectively, but it provides a simple and interpretable way to distinguish more confident and more civically engaged participants from less confident ones.

\subsection{Expertise Modeling}
\label{app: expertise}

Let $e_i$ denote the expertise score of voter $i$. The score is a weighted linear combination of normalized responses to four survey items, two from the exit survey and two from the entry survey~\cite{pournaras2025upgrading, stadtidee2025data}.
\begin{enumerate}
    \item Exit survey item Q20.1, \emph{``I can understand and assess important political issues on the local level well''}, with weight 0.4.
    \item Exit survey item Q21, \emph{``How strongly do you feel connected to Aarau?''}, with weight 0.2.
    \item Entry survey items A02.1 to A02.4, which ask how often the respondent interacts with four groups of persons, each with weight 0.05.
    \begin{itemize}
        \item A02.1, members of the city council
        \item A02.2, members of the residents council
        \item A02.3, members of the city administration
        \item A02.4, other inhabitants of Aarau
    \end{itemize}
    \item Entry survey item F04.2, \emph{``I have the confidence to actively participate in a conversation about local politics''}, with weight 0.2.
\end{enumerate}

All survey responses are linearly normalized to $[0,1]$ before aggregation. The final expertise score is given by
\[
e_i = 0.4 q_{1,i} + 0.2 q_{2,i} + 0.2 q_{4,i} + 0.05(q_{3.1,i} + q_{3.2,i} + q_{3.3,i} + q_{3.4,i}),
\]
where $q_{k,i}$ denotes the normalized response of voter $i$ to item $k$ in the list above, and $q_{3.1,i}$ to $q_{3.4,i}$ are the responses to items A02.1 to A02.4.

The selection of these specific survey items and their relative weights is theoretically grounded in the political science concepts of internal political efficacy and social capital. Internal political efficacy, the perception individuals have of their own ability to understand and participate in political processes, is a well-established and primary predictor of political knowledge and active participation~\cite{niemi1991measuring}. Accordingly, we assign the highest combined weight, 60\%, to direct self-assessments of political understanding in item Q20.1 and communicative confidence in item F04.2. The remaining 40\% of the weight captures social capital and civic embeddedness, which are critical for local municipal expertise~\cite{putnam2000bowling}. A strong connection to the municipality in item Q21 and frequent interactions with local political actors and other inhabitants in items A02.1 to A02.4 indicate that a voter is actively embedded in the information network of the community. Such embedding grants access to localized knowledge that an isolated voter would lack.

\subsection{Expert Selection}

Let $N_{\mathrm{d}} \subseteq N$ be the subset of delegators, the low expertise group, and $N_{\mathrm{r}} \subseteq N$ be the subset of eligible delegates, the high expertise group. Each delegator $i \in N_{\mathrm{d}}$ selects a delegate $j \in N_{\mathrm{r}}$ with probability proportional to a selection weight $w_{ij}$. This weight combines three simple signals that make a delegate more recognizable or more plausible to the delegator. The sealed regime draws delegates independently from this weight. The transparent regime draws them sequentially and scales the weight by the popularity of each delegate.

\subsubsection{Independent Expert Selection}
\label{app: ind_exp_sel}

To model delegate selection by similarity, we extract three key demographic and ideological variables from the survey data, namely age, geographic location, and political orientation. Geographic proximity is determined using the self-reported residential districts of respondents, which are categorized into 18 distinct municipal neighborhoods such as Innenstadt, Zelgli, and Altstadt. Age proximity is calculated directly from self-reported age in years. Finally, ideological alignment is measured using a standard 11-point political self-placement scale, where respondents position themselves on a scale from 0 for ``left'' to 10 for ``right''.

The use of these three specific variables is strongly motivated by the sociological principle of homophily, the well-documented tendency of individuals to associate and bond with similar others~\cite{mcpherson2001birds}. In political and social networks, shared geographic context, generational alignment, and ideological proximity have repeatedly been shown to be primary determinants of trust and interpersonal information flow. By incorporating these factors, our model realistically captures how voters identify relatable, trustworthy delegates based on shared identity and local interests. 

\medskip
\noindent
(i) Geographic proximity.\quad
Voters residing in the same administrative district may share local knowledge and local interests. With $s_i$ the district of voter $i$, we apply a discrete bonus:
\begin{equation}
    \delta_{ij} =
    \begin{cases}
        \beta_g & \text{if } s_i = s_j, \\
        1 & \text{otherwise,}
    \end{cases}
    \label{eq:district}
\end{equation}
where $\beta_g > 1$ is a tunable district bonus parameter, with default value $\beta_g = 2$.

\medskip
\noindent
(ii) Age proximity.\quad
Age is used as a simple proxy for generational closeness. With $y_i$ the age of voter $i$, we apply an exponential decay over the absolute age difference:
\begin{equation}
    \alpha_{ij} = \exp\!\left(-\frac{|y_i-y_j|}{\sigma_a}\right),
    \label{eq:age}
\end{equation}
where $\sigma_a > 0$ controls how quickly the score decreases as age difference grows. The default is $\sigma_a = 20$ years.

\medskip
\noindent
(iii) Political proximity.\quad
Voters may also prefer delegates with similar political orientation. Using self-reported left-right position $p_i$, we apply the analogous decay:
\begin{equation}
    \pi_{ij} = \exp\!\left(-\frac{|p_i-p_j|}{\sigma_p}\right),
    \label{eq:politics}
\end{equation}
where $\sigma_p > 0$ controls tolerance for political distance, with default value $\sigma_p = 2$. Respondents who answered ``don't know'' are assigned $\pi_{ij}=1$ for all $j$, so this term contributes no information.

\medskip
\noindent
The three components are multiplied to obtain the final selection weight:
\begin{equation}
    w_{ij} = \delta_{ij}\alpha_{ij}\pi_{ij}.
    \label{eq:affinity}
\end{equation}
Each delegator $i$ then draws a delegate from $N_{\mathrm{r}}$ according to the normalized distribution, where $i \to j$ denotes the event that $i$ delegates to $j$:
\begin{equation}
    \Pr[\,i \to j\,] = \frac{w_{ij}}{\sum_{k\in N_{\mathrm{r}}} w_{ik}}.
    \label{eq:selection}
\end{equation}

\subsubsection{Sequential Delegation with Herding}
\label{app: seq_del_her}

In the herding variant, non-experts delegate sequentially in a uniformly random order. Each voter $i$, arriving at position $\tau$ in the sequence, can observe the running delegation count $c_j(\tau)$, that is, the number of delegators who have already chosen delegate $j$ before the turn of $i$. Popularity is incorporated directly into the selection weight through the scaling
\begin{equation}
    \tilde{w}_{ij}(\tau) = w_{ij}\bigl(1+c_j(\tau)\bigr)^h,
    \label{eq:herding}
\end{equation}
where $h \geq 0$ is the herding strength parameter. The delegate is then drawn from the normalized distribution over $\tilde{w}_{ij}(\tau)$.

This formulation subsumes the baseline case. At $h=0$, delegation counts have no influence and the rule reduces to Equation~\eqref{eq:selection}. As $h$ increases, popularity amplifies the underlying selection weight, with $h=1$ corresponding to a linear multiplicative interaction between private choice and visible popularity, with $h>1$ causing popularity to dominate more strongly. We use a continuous exponent instead of a binary switch between independent and herding behavior, since social influence typically grows gradually~\cite{lorenz2011social}.

The selection rule in Equation~\eqref{eq:herding} is a preferential attachment kernel over the running delegation counts. Its exponent is in principle measurable from timestamped delegation arrivals using standard kernel estimation methods~\cite{jeong2003measuring, pham2015pafit}. No such estimate exists for a delegation network, and the preferential delegation model of G\"olz et al.\ likewise treats the exponent as a free parameter~\cite{golz2018fluid}. Measurements in adjacent evolving networks bound the plausible range. Citation networks and the Internet grow with near-linear attachment, at exponents of $0.95$ and $1.05$, and collaboration and actor networks at about $0.8$~\cite{jeong2003measuring}. A study of $47$ online social networks finds exponents between roughly $0$ and $1.5$, with $70\%$ of networks sublinear, $30\%$ superlinear, and none exactly linear~\cite{kunegis2013preferential}. Reading against these measurements, our default $h=0.3$ is conservative, $h=1$ corresponds to the near-linear regime typical of citation and Internet growth, and $h=1.5$ is a stress case in the superlinear regime. Publicly recorded on-chain delegation events in DAO governance provide exactly the arrival data the estimators need, and estimating an empirically anchored $h$ from these records is a concrete direction for future work. In this paper we treat $h$ as a free parameter within this bracketed range and report results across it.

\end{document}